%% file: BE_boundeddomains_SR_MD.tex
\documentclass[12pt]{amsart}  

\usepackage{graphicx}
\usepackage{here} 
\usepackage{array} 

\usepackage{vmargin}
\usepackage{amssymb}
\usepackage{mathrsfs}
\usepackage[all]{xy}
\usepackage[usenames,dvipsnames]{color}
\usepackage{soul}
\RequirePackage[colorlinks,linkcolor=blue,citecolor=LimeGreen,urlcolor=red]{hyperref} 
\usepackage{amsmath}

\newtheorem{theorem}{Theorem}[section]

\newtheorem{cor}[theorem]{Corollary}

\newtheorem{defi}[theorem]{Definition}

\newtheorem{lemma}[theorem]{Lemma}

\newtheorem{prop}[theorem]{Proposition}
\newtheorem{remark}[theorem]{Remark}

\numberwithin{equation}{section}

\newcommand{\R}{\mathbb{R}}

\newcommand{\N}{\mathbb{N}}

\newcommand{\func}[3]{#1 : #2 \longrightarrow #3}

\newcommand{\disp}{\displaystyle}
\newcommand{\abs}[1]{\left|#1\right|}
\newcommand{\eps}{\varepsilon}
\newcommand{\norm}[1]{\left\|#1\right\|}

\renewcommand{\leq}{\leqslant}
\renewcommand{\geq}{\geqslant}
\renewcommand{\bar}{\overline}
\renewcommand{\tilde}{\widetilde}
\newcommand{\pa}[1]{\left(#1\right)}
\newcommand{\cro}[1]{\left[#1\right]}
\newcommand{\br}[1]{\left\{#1\right\}}
\newcommand\restr[2]{{
  \left.\kern-\nulldelimiterspace 
  #1 
  \right|_{#2} 
  }}

\def\signmb{\bigskip \begin{center} {\sc
Marc Briant\par\vspace{3mm}
Brown University\par
Division of Applied Mathematics\par
182 George Street, Box F
Providence, RI 02192, USA\par
\vspace{3mm}
e-mail:} \tt{briant.maths@gmail.com} \end{center}}

\begin{document} 

\title[Perturbative theory for the Boltzmann equation in bounded domains]{Perturbative theory for the Boltzmann equation in bounded domains with different boundary conditions}
\author{Marc Briant}
\thanks{The author was supported by the $150^{th}$ Anniversary Postdoctoral Mobility Grant of the London Mathematical Society. The author would also like to acknowledge the Division of Applied Mathematics at Brown University, where this work was achieved.}

\begin{abstract}
We study the Boltzmann equation near a global Maxwellian in the case of bounded domains. We consider the boundary conditions to be either specular reflections or Maxwellian diffusion. Starting from the reference work of Guo \cite{Gu6} in $L^\infty_{x,v}\pa{\pa{1+\abs{v}}^\beta e^{\abs{v}^2/4}}$, we prove existence, uniqueness, continuity and positivity of solutions for less restrictive weights in the velocity variable; namely, polynomials and stretch exponentials. The methods developed here are constructive.
\end{abstract}

\maketitle

\vspace*{10mm}


\smallskip
\textbf{Acknowledgements:} I would like to thank Yan Guo for the fruitful discussions we had.

\tableofcontents

\input{introduction}

\input{mainresults}

\input{collisionfrequencysemigroup}

\input{L2Linftytheory}

\input{extensionstretchpoly}

\input{mainproofs}



%
\bibliographystyle{acm}
\bibliography{bibliography}


\bigskip
\signmb

\end{document}

%% file: introduction.tex
\section{Introduction} \label{sec:intro}

The Boltzmann equation rules the dynamics of rarefied gas particles moving in a domain $\Omega$ of $\R^3$ with velocities in $\R^3$ when the sole interactions taken into account are elastic binary collisions. More precisely, the Boltzmann equation describes the time evolution of $F(t,x,v)$, the distribution of particles in position and velocity, starting from an initial distribution $F_0(x,v)$. It reads
\begin{eqnarray}
\forall t \geq 0 &,& \:\forall (x,v) \in \Omega \times \R^3,\quad  \partial_t F + v\cdot \nabla_x F = Q(F,F),\label{BE}
\\ && \:\forall (x,v) \in \Omega \times \R^3,\quad F(0,x,v) = F_0(x,v). \nonumber
\end{eqnarray}
To which one have to add boundary conditions on $F$. We decompose the phase space boundary
$$\Lambda = \partial \Omega\times\R^3$$
into three sets
\begin{eqnarray*}
\Lambda^+ &=& \br{\pa{x,v}\in\partial\Omega\times\R^3, \quad n(x)\cdot v >0},
\\\Lambda^- &=& \br{\pa{x,v}\in\partial\Omega\times\R^3, \quad n(x)\cdot v <0},
\\\Lambda_0 &=& \br{\pa{x,v}\in\partial\Omega\times\R^3, \quad n(x)\cdot v =0},
\end{eqnarray*}
where $n(x)$ the outward normal at a point $x$ on $\partial\Omega$. The set $\Lambda_0$ is called the grazing set.
\par In the present work, we will consider two types of interactions with the boundary of the domain $\partial\Omega$. Either the specular reflections
\begin{equation}\label{SR}
\forall t > 0,\:\forall (x,v) \in \Lambda^-,\quad  F(t,x,v) = F(t,x,\mathcal{R}_x(v))
\end{equation}
where $\mathcal{R}_x$ stands for the specular reflection at the point $x$ on the boundary:
$$\forall v \in \R^3,\quad \mathcal{R}_x(v) = v - 2(v\cdot n(x))n(x).$$
This interaction describes the fact that the gas particles elastically collide against the wall like billiard balls. The second type is the Maxwellian diffusion boundary condition
\begin{equation}\label{MD}
\forall t>0,\:\forall (x,v) \in \Lambda^-, \quad F(t,x,v)= c_\mu \mu(v)\left[\int_{v_*\cdot n(x)>0} F(t,x,v_*)\left(v_*\cdot n(x)\right)\:dv_*\right]
\end{equation}
where
$$\mu(v) = \frac{1}{\pa{2\pi}^{3/2}}e^{-\frac{\abs{v}^2}{2}} \quad\mbox{and}\quad c_\mu\int_{v\cdot n(x)>0} \mu(v)\left(v\cdot n(x)\right)\:dv=1.$$
This boundary condition expresses the physical process where particles are absorbed by the wall and then emitted back into $\Omega$ according to the thermodynamical equilibrium distribution between the wall and the gas.

\bigskip
The operator $Q(F,F)$ encodes the physical properties of the interactions between two particles. This operator is quadratic and local in time and space. It is given by 
$$Q(F,F) =  \int_{\R^3\times \mathbb{S}^{2}}B\left(|v - v_*|,\mbox{cos}\:\theta\right)\left[F'F'_* - FF_*\right]dv_*d\sigma,$$
where $F'$, $F_*$, $F'_*$ and $F$ are the values taken by $F$ at $v'$, $v_*$, $v'_*$ and $v$ respectively. Define:
$$\left\{ \begin{array}{rl} \displaystyle{v'} & \displaystyle{= \frac{v+v_*}{2} +  \frac{|v-v_*|}{2}\sigma} \vspace{2mm} \\ \vspace{2mm} \displaystyle{v' _*}&\displaystyle{= \frac{v+v_*}{2}  -  \frac{|v-v_*|}{2}\sigma} \end{array}\right., \: \mbox{and} \quad \mbox{cos}\:\theta = \langle \frac{v-v_*}{\abs{v-v_*}},\sigma\rangle .$$
We recognise here the conservation of kinetic energy and momentum when two particles of velocities $v$ and $v_*$ collide to give two particles of velocities $v'$ and $v'_*$.
\par The collision kernel $B$ contains all the information about the interaction between two particles and is determined by physics. We mention, at this point, that one can derive this type of equations from Newtonian mechanics at least formally \cite{Ce}\cite{CIP}. The rigorous validity of the Boltzmann equation from Newtonian laws is known for short times (Landford's theorem \cite{La} or more recently \cite{GST,PSS}).

\bigskip
In the present paper we are interested in the well-posedness of the Boltzmann equation $\eqref{BE}$ for fluctuations around the global equilibrium 
$$\mu(v) = \frac{1}{\pa{2\pi}^{3/2}}e^{-\frac{\abs{v}^2}{2}}.$$
More precisely, in the perturbative regime $F=\mu + f$ we construct a Cauchy theory in $L^\infty_{x,v}$ spaces endowed with strech exponential or polynomial weights and study the continuity and the positivity of such solutions for both specular reflections and diffusive boundary conditions.
\par Under the perturbative regime, the Cauchy problem amounts to solving the perturbed Boltzmann equation
\begin{equation}\label{perturbedBE}
\partial_t f + v\cdot\nabla_x f = Lf + Q(f,f)
\end{equation}
with $L$ being the linear Boltzmann operator
$$Lf = 2Q(\mu,f)$$
and we considered $Q$ as a symmetric bilinear operator
\begin{equation}\label{Qfg}
Q(f,g) = \frac{1}{2}\int_{\R^3\times \mathbb{S}^{2}}B\left(|v - v_*|,\mbox{cos}\:\theta\right)\left[f'g'_* + g'f'_* - fg_*-gf_*\right]dv_*d\sigma.
\end{equation}
Throughout this paper we deal with the perturbed Boltzmann equation $\eqref{perturbedBE}$ and the domain $\Omega$ is supposed to be $C^1$ so that its outwards normal is well-defined (it will be analytic and strictly convex in the case of specular reflections or just connected in the case of Maxwellian diffusion).
\bigskip


\subsection{Notations and assumptions}\label{subsec:notations}

We describe the assumptions and notations we shall use throughout the sequel.

\bigskip
\textbf{Function spaces.} 
Define
$$\langle \cdot \rangle = \sqrt{1+\abs{\cdot}^2}.$$
\par The convention we choose is to index the space by the name of the concerned variable so we have, for $p$ in $[1,+\infty]$,
$$L^p_{[0,T]} = L^p\pa{[0,T]},\quad L^p_{t} = L^p \left(\R^+\right),\quad L^p_x = L^p\left(\Omega\right), \quad L^p_v = L^p\left(\R^3\right).$$
\par For $\func{m}{\R^3}{\R^+}$ a strictly positive measurable function we define the following weighted Lebesgue spaces by the norms
\begin{eqnarray*}
\norm{f}_{L^\infty_{x,v}\pa{m}} &=& \sup\limits_{(x,v)\in\Omega\times\R^3}\cro{\abs{f(x,v)}\:m(v)}
\\\norm{f}_{L^1_vL^\infty_{x}\pa{m}} &=& \int_{\R^3}\sup\limits_{x\in\Omega}\abs{f(x,v)}\:m(v) \:dv
\end{eqnarray*}
and in general with $p$, $q$ in $[1,\infty)$: $\norm{f}_{L^p_vL^q_x\pa{m}} = \norm{\norm{f}_{L^q_x}m(v)}_{L^p_v}$.
\par We define the Lebesgue spaces on the boundary:
\begin{eqnarray*}
\norm{f}_{L^\infty_\Lambda\pa{m}} &=&\sup\limits_{(x,v)\in\Lambda}\cro{\abs{f(x,v)}\:m(v)}
\\\norm{f}_{L^1L^\infty_\Lambda\pa{m}} &=&\int_{\R^3}\sup\limits_{x:\: (x,v)\in\Lambda}\abs{f(x,v)v\cdot n(x)}\:m(v) \:dv
\end{eqnarray*}
with obvious equivalent definitions for $\Lambda^{\pm}$ or $\Lambda_0$. However, when we do not consider the $L^\infty$ setting in the spatial variable we define
$$\norm{f}_{L^2_{\Lambda}\pa{m}} = \cro{\int_{\Lambda} \abs{f(x,v)^2 m(v)^2 \abs{v\cdot n(x)}\:dS(x)dv}}^{1/2},$$
where $dS(x)$ is the Lebesgue measure on $\partial\Omega$.
\par For a function $g$ defined on a space $\mathcal{E}$ and a subset $E \subset \mathcal{E}$ we denote by $\restr{g}{E}$ the restriction of $g$ on $E$.

\bigskip
\textbf{Assumptions on the collision kernel.}
We assume that the collision kernel $B$ can be written as
\begin{equation}\label{B}
B(v,v_*,\theta) = \Phi\left(|v - v_*|\right)b\left( \mbox{cos}\:\theta\right),
\end{equation}
which covers a wide range of physical situations (see for instance \cite{Vi2} Chapter $1$).
\par Moreover, we will consider only kernels with hard potentials, that is 
\begin{equation}\label{hardpot}
\Phi(z) = C_\Phi z^\gamma \:,\:\: \gamma \in [0,1],
\end{equation}
where $C_\Phi>0$ is a given constant. Of special note is the case $\gamma=0$ which is usually known as Maxwellian potentials.
We will assume that the angular kernel $b\circ \mbox{cos}$ is positive and continuous on $(0,\pi)$, and that it satisfies a strong form of Grad's angular cut-off:
\begin{equation}\label{cutoff}
b_\infty=\norm{b}_{L^\infty_{[-1,1]}}<\infty
\end{equation}
The latter property implies the usual Grad's cut-off \cite{Gr1}:
\begin{equation}\label{lb}
l_b = \int_{\mathbb{S}^{d-1}}b\left(\mbox{cos}\:\theta\right)d\sigma = \left|\mathbb{S}^{d-2}\right|\int_0^\pi b\left(\mbox{cos}\:\theta\right) \mbox{sin}^{d-2}\theta \:d\theta < \infty.
\end{equation}
Such requirements are satisfied by many physically relevant cases. The hard spheres case ($b=\gamma=1$) is a prime example.
\bigskip


\subsection{Our goals, strategies and comparison with previous studies} \label{subsec:strategy}

Few results have been obtained about the perturbative theory for the Boltzmann equation with other boundary conditions than the periodicity of the torus. On the torus we can mention \cite{Uk}\cite{Gu3}\cite{Gu4}\cite{MN}\cite{Bri1}\cite{GMM} for collision kernels with hard potentials with cutoff, \cite{GreStr} without the assumption of angular cutoff or \cite{Gu1}\cite{Kim} for soft potentials. A good review of the methods and techniques used can be found in the exhaustive \cite{UkYa}.
\par The study of the well-posedness of the Boltzmann equation, as well as the trend to equilibrium, when the spatial domain is bounded with non-periodic boundary conditions is scarce and only focuses on hard potential kernels with angular cutoff. The cornerstone is the work by Guo \cite{Gu6} who established a complete Cauchy theory around a global Maxwellian and prove the exponential convergence to equilibrium in $L^\infty_{x,v}$ with an important weight $\langle v \rangle^\beta\mu(v)^{-1/2}$. The latter weight is quite restrictive and has been required in all the studies so far. This perturbative theory is done in smooth convex domain for Maxwellian diffusion boundary conditions and strictly convex and analytic domains in the case of specular reflections (note that in-flow and bounce-back boundary conditions are also dealt with). The method of Guo is based on an $L^2-L^\infty$ theory, we briefly explain it later, that was then used in \cite{KimYun} (to obtain similar perturbative results around a rotational local Maxwellian in the case of specular reflections) and recently in \cite{EGKM} to deal with non global diffusive boundary conditions in more general domains.
\par To conclude this overview let us mention that unlike the case of the torus where regularity theory in Sobolev spaces is now well established, a recent result by Kim \cite{Kim1} showed that singularities arise at non-convex points on the boundary even around a global Maxwellian. However, we can still recover some weak form of regularity in $\Omega$ is strictly convex \cite{GKTT1} or if the boundary conditions are diffusive \cite{GKTT2}.

\bigskip
As mentioned before, the main goal of the present work is to establish the perturbative well-posedness and exponential trend to equilibrium for the Boltzmann equation with specular reflexion or diffusive boundary conditions in the $L^\infty_{x,v}$ setting with less restrictive weights than the studies mentioned above. More precisely, we shall deal with $L^\infty_{x,v}\pa{m}$ where $m$ is either a stretch exponential or a polynomial instead of $m=\langle v \rangle^\beta\mu(v)^{-1/2}$ with $\beta$ large. There are two main advances in this work. The first one is a study of transport-like equations with diffusive boundary conditions in a mixed setting $L^1_vL^\infty_x$. The second one is a new analytic version of the extension theory of Gualdani, Mischler and Mouhot \cite{GMM} that fits both the boundary conditions and the lack of hypodissipativity of the linear operator.
\par More precisely, the main contribution of our work if to establish a Cauchy theory in more general spaces. The main strategy is to combine a decomposition of the Boltzmann linear operator $L$ into $A+B$ where $B$ will act like a small perturbation of the operator $G_{\nu} = -v\cdot\nabla_x - \nu(v)$ and $A$ has a regularizing effect. This idea comes from the recent work \cite{GMM} for which we develop here an analytic and non-linear version. The regularizing property of the operator $A$ allows us to decompose the perturbative equation into a system of differential equations
\begin{eqnarray}
\partial_t f_1 + v\cdot\nabla_x f_1 &=& Bf_1 + Q(f_1,f_1+f_2)\label{introf1}
\\\partial_t f_2 + v\cdot\nabla_x f_2 &=& Lf_2 + Q(f_2,f_2) + Af_1 \label{introf2}
\end{eqnarray}
where the first equation can be solved in $L^\infty_{x,v}\pa{m}$ and the second is dealt with in $L^\infty_{x,v}\pa{\langle v \rangle^\beta\mu^{-1/2})}$ where the theory of Guo \cite{Gu6} is known to hold.

\par The key ingredient to study $\eqref{introf1}$ is to show that $G_{\nu}$ along with boundary conditions generates a semigroup $S_{G_\nu}(t)$ exponentially decaying in $L^\infty_{x,v}\pa{m}$. The specular reflections and diffusive boundary conditions cannot be treated by the standard semigroup results in bounded domain \cite{BePro} and we adapt the tools developed in \cite{Gu6} to the weights $m$ considered here. We obtain an explicit form for $S_{G_\nu(t)}$ in the case of specular reflection whereas we only have an implicit description of it in the case of Maxwellian diffusion. The latter implicit description includes the contribution of all the possible backward characteristic trajectories starting at $(t,x,v)$. We then use the fact that the measure of the set of trajectories not reaching the initial plane $\br{t=0}$ is small.
\par The second difficulty in solving $\eqref{introf1}$ is to prove that $B$ does not perturb ``too much'' the exponential decay generated by the semigroup $S_{G_\nu}(t)$. Indeed, the latter semigroup is not strongly continuous and we therefore loose the hypodissipativity properties that hold for $G_\nu$ is the case of the torus \cite{GMM}. The case of specular reflections can be dealt with thanks to a Duhamel formulation because $S_{G_\nu}(t)$ has a good contractive property. Such a property is missing in the case of diffusive boundary condition. Due to the implicit description of $S_{G_\nu}(t)$, the proof of $B$ being a small perturbation of $G_\nu$ requires a $L^1_vL^\infty_x$-theory for the semigroup $S_{G_\nu}(t)$ as well as a new mixing estimate for $B$. The study of transport-like equations with boundary conditions in mixed norms seems new to our knowledge.
\par The second equation $\eqref{introf2}$ can be solved easily using the regularizing property of the operator $A$ and the results already described for specular reflections in strictly convex and analytic domains or \cite{EGKM} for Maxwellian diffusion boundary condition in $C^1$ bounded domains.

\bigskip
We conclude by mentioning that our results also give the continuity of the aforementioned solutions
 away from the grazing set $\Lambda_0$. Such a property also allows us to obtain the positivity (and quantify it explicitely) of the latter solutions thanks to recent results by the author \cite{Bri2}\cite{Bri5}.
\bigskip


\subsection{Organisation of the article}\label{subsec:organization}

Section $\ref{sec:mainresults}$ is dedicated to the statement and the description of the main results proved in this paper. We also give some background properties about the linear Boltzmann operator.
\par In Section $\ref{sec:semigroupcollisionfrequency}$ we study the semigroup generated by the transport part and the collision frequency kernel $G_\nu=-v\cdot\nabla_x -\nu$ along with boundary conditions.
\par We give a brief review of the existing $L^2-L^\infty$ theory for the full linear perturbed operator $G=-v\cdot\nabla_x +L$ in Section $\ref{sec:L2Linftytheory}$.
\par We present and solve the system of equations $\eqref{introf1}$-$\eqref{introf2}$ in Section $\ref{sec:extensionSRMD}$.
\par Lastly, Section $\ref{sec:mainproofs}$ is dedicated to the proof of existence, uniqueness, exponential decay, continuity and positivity of solutions to the full Boltzmann equation $\eqref{BE}$.
\bigskip

%% file: mainresults.tex
\section{Main results} \label{sec:mainresults}

\subsection{Some essential background on the perturbed Boltzmann equation}\label{subsec:background}

We gather here some renown properties about the Boltzmann equation.

\bigskip
\textbf{\textit{A priori} conservation laws.}We start by noticing the symmetry property of the Boltzmann operator (see \cite{Ce}\cite{CIP}\cite{Vi2} among others).

\bigskip
\begin{lemma}\label{lem:integralQ}
Let $f$ be such that $Q(f,f)$ is well-defined. Then for all $\Psi(v)$ we have
$$\int_{\R^3}Q(f,f)\Psi\:dv = \frac{C_\Phi}{4}\int_{\R^d\times\R^d\times\mathbb{S}^{d-1}}q(f)(v,v_*)\left[\Psi'_* + \Psi' - \Psi_* - \Psi\right]\:d\sigma dvdv_*,$$
with
$$q(f)(v,v_*) = |v-v_*|^\gamma b\left(\mbox{cos}\:\theta\right)ff_*.$$
\end{lemma}
\bigskip

This result is well-known for the Boltzmann equation and is a simple manipulation of the integrand using changes of variables $(v,v_*)\to (v_*,v)$ and $(v,v_*)\to (v',v'_*)$, as well as using the symmetries of the operator $q(f)$. A straightforward consequence of the above is the \textit{a priori} conservation of mass when one consider either specular reflections or Maxwellian diffusion
\begin{equation}\label{conservationmass}
\forall t\geq 0, \quad \int_{\Omega\times\R^3} f(t,x,v)\:dxdv = \int_{\Omega\times\R^3} f_0(x,v)\:dxdv.
\end{equation}

\par In the case of specular reflections Lemma $\ref{lem:integralQ}$ also implies the \textit{a priori} conservation of energy
\begin{equation}\label{conservationenergy}
\forall t\geq 0, \quad \int_{\Omega\times\R^3}\abs{v}^2 f(t,x,v)\:dxdv = \int_{\Omega\times\R^3} \abs{v}^2f_0(x,v)\:dxdv.
\end{equation}

\par Lastly, in the specific case of specular reflections inside a domain $\Omega$ with an axis of rotation symmetry:
\begin{equation}\label{axissymmetric}
\exists x_0,\:\omega \in \R^3, \: \forall x\in\partial\Omega, \quad \br{(x-x_0)\times\omega}\cdot n(x)=0, 
\end{equation}
we also obtain the \textit{a priori} conservation of the following angular momentum
\begin{equation}\label{conservationangularmomentum}
\forall t\geq 0, \: \int_{\Omega\times\R^3} \br{(x-x_0)\times\omega}\cdot v f(t,x,v)\:dxdv = \int_{\Omega\times\R^3} \br{(x-x_0)\times\omega}\cdot vf_0(x,v)\:dxdv.
\end{equation}

\bigskip
\textbf{The linear Boltzmann operator.} We gather some well-known properties of the linear Boltzmann operator $L$ (see \cite{Ce, CIP, Vi2, GMM} for instance). 
\par $L$ is a closed self-adjoint operator in $L^{2}_v\left(\mu^{-1/2}\right)$ with kernel
$$\mbox{Ker}\left(L\right) = \mbox{Span}\left\{\phi_0(v),\dots,\phi_{4}(v)\right\}\mu ,$$
where $\pa{\phi_i}_{0\leq i\leq 4}$ is an orthonormal basis of $\mbox{Ker}\left(L\right)$ in $L^2_v\left(\mu^{-1/2}\right)$. More precisely, if we denote $\pi_L$ the orthogonal projection onto $\mbox{Ker}\left(L\right)$ in $L^2_v\left(\mu^{-1/2}\right)$:
\begin{equation}\label{piL}
\left\{\begin{array}{l} \disp{\pi_L(g) = \sum\limits_{i=0}^{4} \pa{\int_{\R^3} g(v_*)\phi_i(v_*)\:dv_*} \phi_i(v)\mu(v)} \vspace{2mm}\\\vspace{2mm} \disp{\phi_0(v)=1,\quad \phi_i(v) = v_i,\:1\leq i \leq 3,\quad \phi_{4}(v)=\frac{\abs{v}^2-3}{\sqrt{6}},}\end{array}\right.
\end{equation}
and we define $\pi_L^\bot = \mbox{Id} - \pi_L$. The projection $\pi_L(f(x,\cdot))(v)$ of $f(x,v)$ onto the kernel of $L$ is called its fluid part whereas $\pi_L^\bot(f)$ is its microscopic part.
\par $L$ can be written under the following form
\begin{equation}\label{LnuK}
L = -\nu(v) +K,
\end{equation}
where $\nu(v)$ is the collision frequency
$$\nu(v) = \int_{\R^3\times\mathbb{S}^{2}} b\left(\mbox{cos}\:\theta\right)\abs{v-v_*}^\gamma \mu_*\:d\sigma dv_*$$
and $K$ is a bounded and compact operator in $L^2_v\pa{\mu^{-1/2}}$ that takes the form
$$K(f)(v) = \int_{\R^3} k(v,v_*)f(v_*)\:dv_*.$$
\par Finally we remind that there exist $\nu_0,\:\nu_1 >0$ such that
\begin{equation}\label{nu0nu1}
\forall v \in \R^3,\quad \nu_0(1+\abs{v}^\gamma)\leq \nu(v)\leq \nu_1(1+\abs{v}^\gamma),
\end{equation}
and that $L$ has a spectral gap $\lambda_L >0$ in $L^2_{x,v}\pa{\mu^{-1/2}}$ (see \cite{BM,Mo1} for explicit proofs)
\begin{equation}\label{spectralgapL}
\forall f \in L^2_v\pa{\mu^{-1/2}}, \quad \langle L(f),f\rangle_{L^2_v\pa{\mu^{-1/2}}} \leq -\lambda_L \norm{\pi_L^\bot(f)}_{L^2_v\pa{\mu^{-1/2}}}^2.
\end{equation}

\bigskip
\textbf{The linear perturbed Boltzmann operator.}
The linear perturbed Boltzmann operator is the full linear part of the perturbed Boltzmann equation $\eqref{perturbedBE}$:
$$G = L - v\cdot \nabla_x.$$

An important point is that the same computations as to show the \textit{a priori} conservation laws implies that in $L^2_{x,v}\pa{\mu^{-1/2}}$ the space $\pa{\mbox{Span}\br{\mu,\abs{v}^2\mu}}^\bot$ is stable under the flow
$$\partial_t f = G(f)$$
with specular reflections whereas $\pa{\mbox{Span}\br{\mu}}^\bot$ is stable under the same differential equation with diffusive boundary conditions. We thus define the $L^2_{x,v}\pa{\mu^{-1/2}}$-projection onto that space
\begin{equation}\label{PiGSR}
\Pi_G(f)(v) = \left(\int_{\Omega\times\R^3} h(x,v_*)\:dxdv_*\right)\mu(v) + \left(\int_{\Omega\times\R^3} \abs{v_*}^2h(x,v_*)\:dxdv_*\right)\abs{v}^2\mu(v),
\end{equation}
(with the addition of the angular momentum term when $\Omega$ is axis-symmetric) and in the case of Maxwellian diffusion
\begin{equation}\label{PiGMD}
\Pi_G(f)(v) = \left(\int_{\Omega\times\R^3} h(x,v_*)\:dxdv_*\right)\mu(v).
\end{equation}
Again we define $\Pi_G^\bot = \mbox{Id} - \Pi_G$.

\bigskip
In order to avoid repeating the conservation laws, for a function space $E$ we define the following sets
\begin{eqnarray*}
\mbox{SR}\cro{E} &=& \br{f\in E,\quad \Pi_G(f) = 0,\:\Pi_G \:\mbox{defined for specular reflection}\:\eqref{PiGSR}}
\\\mbox{MD}\cro{E} &=& \br{f\in E,\quad \Pi_G (f)=0,\:\Pi_G \:\mbox{defined for specular reflection}\:\eqref{PiGMD}}.
\end{eqnarray*}
This amounts to saying that the functions in $\mbox{SR}\cro{E}$ satisfy the conservation of mass $\eqref{conservationmass}$ and energy $\eqref{conservationenergy}$ (and angular momentum $\eqref{conservationangularmomentum}$ if $\Omega$ is axis-symmetric) whilst the functions in $\mbox{MD}\cro{E}$ satisfy the conservation of mass $\eqref{conservationmass}$.
\bigskip


\subsection{Main theorems}\label{subsec:maintheorems}

We start with the following definition.
\bigskip
\begin{defi}
Let $\Omega$ be a bounded domain in $\R^3$. We say that $\Omega$ is analytic and strictly convex if there exists an analytic function $\func{\xi}{\R^3}{\R}$ such that $\Omega = \br{x:\:\xi(x)<0}$ and
\begin{itemize}
\item at the boundary $\xi(x)=0$ and $\nabla\xi(x)\neq 0$,
\item there exists $c_\xi>0$ such that for all $x\in\R^3$, 
\begin{equation}\label{strictlyconvex}
\sum\limits_{1\leq i,j \leq 3}\partial_{ij}\xi(x)x_ix_j \geq c_\xi\abs{x}^2.
\end{equation}
\end{itemize}
\end{defi}
\bigskip

The present work is dedicated to proving the following two perturbative studies for the Boltzmann equation in bounded domains.

\bigskip
\begin{theorem}\label{theo:cauchySR}
Let $\Omega$ be an analytic strictly convex $\eqref{strictlyconvex}$ bounded domain and $B$ be a collision kernel of the form $\eqref{B}$ with hard potential $\eqref{hardpot}$ and angular cutoff $\eqref{cutoff}$.   Let $m=e^{\kappa\abs{v}^\alpha}$ with $\kappa >0$ and $\alpha$ in $(0,2)$ or $m=\langle v \rangle^k$ with 
$$k>1+\gamma +\frac{16\pi b_\infty}{l_b}$$
where $b_\infty$ and $l_b$ were defined by $\eqref{cutoff}$ and $\eqref{lb}$.
\\ Then there exists $\eta_0$, $C_0$ and $\lambda_0 >0$ such that if $F_0 = \mu + f_0$ with $f_0$ in $\mbox{SR}\cro{L^\infty_{x,v}\pa{m}}$ satisfies
$$\norm{f_0}_{L^\infty_{x,v}\pa{m}}\leq \eta_0$$
then there exists a unique $F = \mu + f$ with $f$ in $L^\infty_{[0,+\infty)}\pa{\mbox{SR}\cro{L^\infty_{x,v}\pa{m}}}$ solution to the Boltzmann equation $\eqref{BE}$ with specular reflections boundary conditions $\eqref{SR}$. Moreover, the following holds
\begin{enumerate}
\item $$\forall t \geq 0, \quad \norm{f(t)}_{L^\infty_{x,v}\pa{m}} \leq C_0 e^{-\lambda_0 t}\norm{f_0}_{L^\infty_{x,v}\pa{m}};$$
\item if $F_0 \geq 0$ is continuous on $\bar{\Omega}\times\R^3-\Lambda_0$ and satisfies the specular reflections boundary condition then $F \geq 0$ and $F$ is continuous on $[0,+\infty) \times \pa{\bar{\Omega}\times\R^3-\Lambda_0}$.
\end{enumerate}
\end{theorem}
\bigskip

\begin{remark}
We make a few comments about the previous result.
\begin{itemize}
\item The analyticity and the strict convexity of the domain are required to ensure that one can use the  control of the $L^\infty_{x,v}\pa{\langle v \rangle^\beta \mu^{-1/2}}$ theory by the $L^2_{x,v}\pa{\mu^{-1/2}}$ developed in \cite{Gu6} (see Remark \ref{rem:analytic}). Moreover, the methods are constructive starting from \cite{Gu6}. The constants are thus not explicit since the methods in \cite{Gu6} are not. Obtaining a constructive theory in the latter spaces, thus getting rid of the strong assumption of analyticity, would be of great interest;
\item The positivity of $F$ is actually quantified \cite{Bri2}\cite{Bri5} and is an explicit Maxwellian lower bound. We refer to Subsection $\ref{subsec:positivity}$ for more details;
\item The uniqueness is obtained in a perturbative setting, \textit{i.e.} on the set of function of the form $F = \mu +f$ with $f$ small. If the uniqueness of solutions to the Boltzmann equation in $L^1_vL^\infty_x\pa{\langle v \rangle^{2+0}}$ is known on the torus \cite{GMM} a uniqueness theory outside the perturbative regime remains, at this date, an open problem in the case of bounded domains. 
\end{itemize}
\end{remark}
\bigskip

We obtain a similar result in the case of Maxwellian diffusion boundary condition. As explained in the introduction, the assumptions on the domain $\Omega$ are far less restrictive. We however define two new sets. 
\par For $(x,v)$ define the backward exit time by $t_{b}(x,v) = \inf\br{t>0,\: x-tv \notin \Omega}$ and the footprint $x_b(x,v) = x-t_b(x,v)v$. Define the singular grazing boundary 
$$\Lambda_0^{(S)} = \br{(x,v)\in \Lambda_0, \quad t_b(x,v)\neq 0 \:\mbox{or}\:t_b(x,-v)\neq 0 }$$
and the discontinuity set
$$\mathfrak{D} = \Lambda_0 \cup \br{(x,v)\in\bar{\Omega}\times\R^3, \quad (x_b(x,v),v)\in\Lambda_0^{(S)}}$$

\bigskip
\begin{theorem}\label{theo:cauchyMD}
Let $\Omega$ be a $C^1$ connected bounded domain and $B$ be a collision kernel of the form $\eqref{B}$ with hard potential $\eqref{hardpot}$ and angular cutoff $\eqref{cutoff}$.   Let $m=e^{\kappa\abs{v}^\alpha}$ with $\kappa >0$ and $\alpha$ in $(0,2)$ or $m=\langle v \rangle^k$ with 
$$k>1+\gamma +\frac{16\pi b_\infty}{l_b}$$
where $b_\infty$ and $l_b$ were defined by $\eqref{cutoff}$ and $\eqref{lb}$.
\\ Then there exists $\eta_0$, $C_0$ and $\lambda_0 >0$ such that if $F_0 = \mu + f_0$ with $f_0$ in $\mbox{MD}\cro{L^\infty_{x,v}\pa{m}}$ satisfies
$$\norm{f_0}_{L^\infty_{x,v}\pa{m}}\leq \eta_0$$
then there exists a unique $F = \mu + f$ with $f$ in $L^\infty_{[0,+\infty)}\pa{\mbox{MD}\cro{L^\infty_{x,v}\pa{m}}}$ solution to the Boltzmann equation $\eqref{BE}$ with Maxwellian diffusion boundary conditions $\eqref{MD}$. Moreover, the following holds
\begin{enumerate}
\item $$\forall t \geq 0, \quad \norm{f(t)}_{L^\infty_{x,v}\pa{m}} \leq C_0 e^{-\lambda_0 t}\norm{f_0}_{L^\infty_{x,v}\pa{m}};$$
\item if $F_0\geq 0$ is continuous on $\bar{\Omega}\times\R^3-\Lambda_0$ and satisfies the Maxwellian diffusion boundary condition then $F\geq 0$ is continuous on $[0,+\infty) \times \pa{\bar{\Omega}\times\R^3-\mathfrak{D}}$.
\end{enumerate}
\end{theorem}
\bigskip

\begin{remark}
We first emphasize that the latter Theorem is obtained with constructive arguments and the constants $\eta_0$, $C_0$ and $\lambda_0 >0$ can be computed explicitly in terms of $m$ and the collision operator. Then we make a few comments.
\begin{itemize}
\item In the case of a convex domain $\mathfrak{D} = \Lambda_0$ (see \cite{EGKM} Lemma $3.1$);
\item The rate of trend to equilibrium $\lambda_0$ can be chosen as close as one wants from the optimal one in the $L^2\pa{\mu^{-1/2}}$ framework;
\item The positivity of $F$ can be quantified in the case $\Omega$ convex \cite{Bri5}. We obtain an explicit Maxwellian lower bound, see Subsection $\ref{subsec:positivity}$;
\item Here again the uniqueness is obtained only in a perturbative setting.
\end{itemize}
\end{remark}
\bigskip

%% file: collisionfrequencysemigroup.tex
\section{Preliminaries: semigroup generated by the collision frequency}\label{sec:semigroupcollisionfrequency}

For general domains $\Omega$, the Cauchy theory in $L^p_{x,v}$ ($1\leq p <+\infty$) of equations of the type
$$\partial_t f + v\cdot\nabla_x f = g$$
with boundary conditions
$$\forall (x,v) \in \Lambda^-, \quad f(t,x,v) = P(f)(t,x,v),$$
where $\func{P}{L^p_{\Lambda^+}}{L^p_{\Lambda^-}}$ is a bounded linear operator, is well-defined in $L^p_{x,v}$ when $\norm{P} <1$ \cite{BePro}. The specific case $\norm{P} = 1$ can still be dealt with (\cite{BePro} Section $4$) but if the existence of solutions in $L^p_{x,v}$ can be proved, the uniqueness is not always given unless one can prove that the trace of $f$ belongs to $L^2_{\mbox{\scriptsize{loc}}}\pa{\R^+;L^p_{x,v}\pa{\Lambda}}$.
\par For specular reflections or Maxwellian diffusion boundary conditions, the boundary operator $P$ is of norm exactly one and the general theory fails. If this generates difficulties for the full linear operator $L$, we can overcome this problem in the case of a mere multiplicative function $g=\nu(v)$.

\bigskip
This section is devoted to proving that the following operator
$$G_\nu = -\nu(v) -v\cdot\nabla_x$$
generates a semigroup $S_{G_\nu}(t)$ in two different frameworks: the specular reflections and the Maxwellian diffusion . We prove that $G_\nu$ along with either specular reflections or Maxwellian diffusion generates a semigroup with exponential decay in $L^\infty_{x,v}$ spaces endowed with polynomial and stretch exponential weights.
\par The $L^\infty_v$ setting is essential for the existence of solutions in the case of specular reflections since one needs to control the solution along the characteristic trajectories (see Remark $\ref{rem:whyLinftyx}$) whereas we show that in the case of diffusion it also generates a semigroup in weighted $L^1_vL^\infty_x$.

\bigskip
\par We emphasize here that such a study was done in \cite{Gu6} in $L^\infty_{x,v}\pa{m(v)\mu^{-1/2}}$. We extend his proofs to more general and less restrictive weights as well as to a new $L^1_vL^\infty_x $ setting in the diffusive setting.
\par Regarding existence and uniqueness, the methods are standard in the study of linear equations in bounded domains \cite{BePro} for specular reflections and rely on an approximation of the boundary operator $P$ ($\norm{P}= 1$). The case of Maxwellian diffusion is different since the norm of the boundary operator heavily depends on the weight function. In the case $L^\infty_{x,v}$ we prove that we can use the arguments developed in \cite{Gu6} whereas the new framework $L^1_vL^\infty_x$ requires new estimates to obtain weak converge which does not come directly from uniform boudnedness in $L^1$.
\par The exponential decay is more intricate and requires a description of the characteristic trajectories for the free transport equation with boundary conditions to obtain explicit formula in terms of $f_0$ for $S_{G_\nu}(t)f_0$. Although this is possible in the case of specular reflections, such an explicit form is not known for the Maxwellian diffusion and it has to be dealt with using equivalent norms.
\bigskip


\subsection{The case of specular reflections}\label{subsec:GnuSR}

As we shall see, the case of specular reflections in a weighted Lebesgue space is equivalent to the same problem with a weight $1$. The study in $L^\infty_{x,v}$ has been done in \cite{Gu6} Lemma $20$ but we write it down for the sake of completeness.

\bigskip
\begin{prop}\label{prop:semigroupGnuSR}
Let $m=e^{\kappa\abs{v}^\alpha}$ with $\kappa >0$ and $\alpha$ in $(0,2)$ or $m=\langle v \rangle^k$ with $k$ in $\N$; let $f_0$ be in $L^\infty_{x,v}\pa{m}$. Then there exists a unique solution $S_{G_\nu}(t)f_0 \in L^\infty_{x,v}\pa{m}$ to 
\begin{equation}\label{eqGnu}
\cro{\partial_t + v\cdot\nabla_x + \nu(v)}\pa{S_{G_\nu}(t)f_0}=0
\end{equation}
such that $\restr{\pa{S_{G_\nu}(t)f_0}}{\Lambda} \in L^\infty_\Lambda\pa{m}$ and satisfying the specular reflections $\eqref{SR}$ with initial data $f_0$. Moreover it satisfies
$$\forall t\geq 0, \quad \norm{S_{G_\nu}(t)f_0}_{L^\infty_{x,v}\pa{m}} \leq e^{-\nu_0t}\norm{f_0}_{L^\infty_{x,v}\pa{m}},$$
with $\nu_0 = \inf\br{\nu(v)}>0$.
\end{prop}
\bigskip

\begin{proof}[Proof of Proposition $\ref{prop:semigroupGnuSR}$]

The proof will be done in three steps: uniqueness, existence and finally the exponential decay.
\par We start by noticing that if $f$ belongs to $L^\infty_{x,v}\pa{m}$ and satisfies
$$\cro{\partial_t + v\cdot\nabla_x + \nu(v)}f(t)=0$$
with specular reflections boundary condition then $h= m(v) f$ is also a solution with specular reflections and $h$ belongs to $L^\infty_{x,v}$ and its restriction on $\Lambda$ belongs to $L^\infty _\Lambda$. Thus, we only prove the proposition in the case $m=1$.

\bigskip
\textbf{Step 1: Uniqueness.} Assume that there exists such a solution $f$ in $L^\infty_{x,v}$.
\par Consider the function $h(t,x,v) = \langle v \rangle^{-\beta} f(t,x,v)$ where $\beta$ is chosen such that
\begin{equation}\label{beta}
\langle v \rangle^{-2\beta} \pa{1+\abs{v}} \in L^1_v.
\end{equation}
A mere Cauchy-Schwarz inequality shows that $h$ is in $L^2_{x,v}$ and $\restr{h}{\Lambda} \in L^2_\Lambda$. Moreover, $h$ satisfies the same differential equality as $f$. Multiply $\eqref{eqGnu}$ by $h$ and integrating in $x$ and $v$, we can use the divergence theorem on $\Lambda$ and the fact that $\nu(v)\geq \nu_0>0$: 
\begin{eqnarray}
\frac{1}{2}\frac{d}{dt}\norm{h}^2_{L^2_{x,v}} &=& \int_{\Omega\times\R^3}h(t,x,v)\cro{-v\cdot\nabla_x - \nu(v)}h(t,x,v)\:dxdv\nonumber
\\&=& -\int_{\Omega\times\R^3} v\cdot\nabla_x\pa{h^2}\:dxdv - \norm{\nu(v)h}^2_{L^2_{x,v}}\nonumber
\\&\leq& -\int_{\Lambda}\abs{h(t,x,v)}^2\pa{v\cdot n(x)}\:dS(x)dv - \nu_0\norm{h}^2_{L^2_{x,v}}.\label{equniquenessL2}
\end{eqnarray}
The integral on $\Lambda$ is null since $h$ satisfies the specular reflections and therefore we can apply a Gr\"onwall lemma to $\norm{h}_{L^2_{x,v}}$ and obtain the uniqueness for $h$ and thus for $f$.

\bigskip
\textbf{Step 2: Existence.}  Existence is proved by approximating the specular reflections in order to get a decrease at the boundary and be in the case $\norm{P} <1$.

\bigskip
Let $f_0$ be in $L^\infty_{x,v}$.
\par For any $\eps$ in $(0,1)$ we consider the following differential problem with $h_\eps \in L^\infty_{x,v}$ and $\restr{h_\eps}{\Lambda} \in L^\infty _\Lambda$:
\begin{equation}\label{eqdiffGnu}
\cro{\partial_t + v\cdot\nabla_x +\nu} h_\eps =0,\quad h_\eps(0,x,v)= f_0(x,v)
\end{equation}
with the absorbed specular reflections boundary condition
$$\forall (t,x,v) \in \R^3\times \Lambda^-, \quad h_\eps(t,x,v) = (1-\eps)h_\eps(t,x,\mathcal{R}_x(v))).$$
This problem has a unique solution. Indeed we construct the following iterative scheme
\begin{equation*}
\left\{\begin{array}{l}\disp{\cro{\partial_t + v\cdot\nabla_x +\nu} h^{(l+1)}_\eps =0,\quad h^{(l+1)}_\eps(0,x,v)= f_0(x,v)}\vspace{2mm} \\  \disp{\restr{h^{(0)}_\eps}{\Lambda^+}=0\quad\mbox{and}\quad\forall t > 0, \:\forall (x,v)\in\Lambda^-,\: h^{(l+1)}_\eps(t,x,v) = h^{(l)}_\eps\pa{t,x,\mathcal{R}_x(v)}.}\end{array}\right.
\end{equation*}
The functions $h^{(l)}_\eps$ are well-defined because the boundary condition is now an in-flow boundary condition to which existence is known (see \cite{Gu6} Lemma $12$ for instance).

\bigskip
We know that along the characteristic trajectories (straight lines in between two rebounds, see \cite{Bri2} Appendix for rigorous construction of characteristics in a $C^1$ bounded domain) $e^{\nu t}h^{(l+1)}_\eps$ is constant. We denote the first backward exit time by
\begin{equation}\label{tmin}
t_{min}(x,v) = \max\br{t\geq 0; \quad x-sv \in \bar{\Omega},\: \forall 0\leq s \leq t}.
\end{equation}

\par Consider $(x,v) \notin \Lambda_0\cup\Lambda^-$, then $t_{min}(x,v)> 0$. If $t_{min}(x,v) \geq t$ then the backward characteristic line starting at $(x,v)$ at time $t$ reaches the initial plan $\br{t=0}$ whereas if $t_{min}(x,v) < t$ it hits the boundary at time $t-t_{min}(x,v)$ at $(x-t_{min}(x,v)v,v)\in\Lambda^-$, where we can apply the boundary condition. Therefore we have the following representation for $h^{(l+1)}_\eps$ for all $t\geq 0$ and for almost all $(x,v) \notin \Lambda_0\cup\Lambda^-$,
\begin{equation}\label{representationGnueps}
\begin{split}
h^{(l+1)}_\eps(t,x,v) =& \mathbf{1}_{t_{min}(x,v) \geq t}\: e^{-\nu(v)t}f_0(x-tv,v) 
\\&+ \mathbf{1}_{t_{min}(x,v) < t}\: (1-\eps)e^{-\nu(v)t_{min}(x,v)}\:\restr{h^{(l)}_\eps}{\Lambda^+}(t-t_{min},x_1,v_1),
\end{split}
\end{equation}
where we defined $x_1 = x-t_{min}(x,v)v$ and $v_1 = \mathcal{R}_{x_1}(v)$.

\bigskip
For all $t\geq 0$, for all $(x,v)\notin\Lambda_0\cup\Lambda^-$ and for all $l\geq 1$,
\begin{equation}\label{cauchysequence}
\abs{h^{(l+1)}_\eps(t,x,v) - h^{(l)}_\eps(t,x,v)} \leq (1-\eps)\Big|\restr{h^{(l)}_\eps}{\Lambda^+}(t,x_1,v_1) - \restr{h^{(l-1)}_\eps}{\Lambda^+}(t,x_1,v_1)\Big|.
\end{equation}
Thus, considering $(x,v)\in \Lambda^+$ we show that $\pa{\restr{h_\eps^{(l)}}{\Lambda^+}}_{l\in\N}$ is a Cauchy sequence in $L^\infty_t L^\infty _{\Lambda^+}$. Then by the boundary condition it implies that $\pa{\restr{h_\eps^{(l)}}{\Lambda^-}}_{l\in\N}$ is also a Cauchy sequence in $L^\infty_t L^\infty _{\Lambda^-}$. Finally from $\eqref{cauchysequence}$, $\pa{h_\eps^{(l)}}_{l\in\N}$ is also a Cauchy sequence in $L^\infty_t L^\infty_{x,v}$.

\bigskip
\begin{remark}\label{rem:whyLinftyx}
The $L^\infty$ framework is essential to obtain the control of $\abs{h^{(l+1)}_\eps - h^{(l)}_\eps}$ by the control of $\abs{h^{(l)}_\eps - h^{(l-1)}_\eps}$ at $(x_1,v_1)$. Any other $L^p_x$ spaces would have required to change the variable $v_1 \mapsto v$ to which a computation of the jacobian is still a very hard problem (see \cite{Gu6} Subsection $4.3.1$) and can be $0$.
\end{remark}
\bigskip

\par We obtain existence of $h_\eps$ solution to $\eqref{eqdiffGnu}$ by letting $l$ tend to infinity. The latter solution is unique since its restriction on the boundary belongs to $L^\infty _{\Lambda}$. Indeed, we can apply the divergence theorem as in $\eqref{equniquenessL2}$ which yields uniqueness because the integral on $\Lambda$ is positive since $1-\eps <1$.

\bigskip
It only remains to show that one can indeed take the limit of $\pa{h_\eps}_{\eps>0}$ when $\eps$ goes to zero.
\par We remind that we chose $\restr{h^{(0)}_\eps}{\Lambda^+}=0$ and therefore by $\eqref{representationGnueps}$ applied to $(x,v) \in \Lambda^+$:
$$\abs{h^{(l+1)}_\eps(t,x,v)} \leq \left\{\begin{array}{l} \disp{\norm{f_0}_{L^\infty_{x,v}} \quad\mbox{if}\quad t\leq t_{min}(x,v)} \vspace{2mm}\\ \disp{\norm{\restr{h^{(l)}_\eps}{\Lambda^+}}_{L^\infty_{\Lambda^+}} \quad\mbox{if}\quad t>t_{min}(x,v).}   \end{array} \right.$$
The latter further implies
\begin{equation}\label{uniformepsLambda+SR}
\forall l\geq 0,\:\forall t \geq 0, \quad \norm{h^{(l)}_\eps(t,\cdot,\cdot)}_{L^\infty _{\Lambda^+}}\leq \norm{f_0}_{L^\infty_{x,v}}.
\end{equation}
The boundary condition then implies
\begin{equation}\label{uniformepsLambda-SR}
\forall l\geq 0,\:\forall t \geq 0, \quad \norm{h^{(l)}_\eps(t,\cdot,\cdot)}_{L^\infty _{\Lambda^-}}\leq \norm{f_0}_{L^\infty_{x,v}},
\end{equation}
and finally the representation of $h^{(l+1)}_\eps$ $\eqref{representationGnueps}$ combined with $\eqref{uniformepsLambda+SR}$ yields
\begin{equation}\label{uniformepsOmega}
\forall t \geq 0,\quad \norm{h^{(l)}_\eps(t,\cdot,\cdot)}_{L^\infty_{x,v}}\leq \norm{f_0}_{L^\infty_{x,v}}.
\end{equation}
\par From the uniform controls $\eqref{uniformepsLambda+SR}-\eqref{uniformepsLambda-SR}-\eqref{uniformepsOmega}$ one can take a weak-* limit of $h_\eps$ in  $L^\infty_{t,x,v}$ and of $\restr{h_\eps}{\Lambda}$ in $L^\infty_tL^\infty_\Lambda$ and such a limit is solution to our initial problem.

\bigskip
\textbf{Step 3: Exponential decay.} We use the study of backwards characteristic trajectories of the transport equation in $C^1$ bounded domains derived in \cite{Bri2}. 
\par For $(x,v) \notin \Lambda_0$, the backwards trajectory starting from $(x,v)$ are straight lines in between two consecutive rebounds. We define a sequence of rebounds $(t_i,x_i,v_i) = (t_i(x,v),x_i(x,v),v_i(x,v))$ with $(t_0,x_0,v_0)=(t,x,v)$ that are the footprints (and time) of the backward trajectories of the transport equation in $\Omega$ starting at $(x,v)$ at time $t$ (see \cite{Bri2} Proposition $A.8$ and Definition $A.6$). Moreover, the sequence $(t_i,x_i,v_i)$ is almost always well defined (countably many rebounds) and finite for any given $t\geq 0$ (see \cite{Bri2} Proposition $A.4$).
\par With this description of characteristics we can iterate the process initiated in $\eqref{representationGnueps}$. This gives that $e^{\nu(v)t}h_\eps$ is constant along characteristics and
\begin{equation}\label{representationGnuSR}
h_\eps(t,x,v) = \sum\limits_{i}\mathbf{1}_{[t_{i+1},t_i)}(0)\cro{1-\eps}^i e^{-\nu(v)t}f_0(x_i-t_iv_i,v_i),
\end{equation}
for almost every $(x,v)\in\bar{\Omega}\times\R^3-\Lambda_0$. Note that we used that $\nu(v_i)=\nu(v)$ because $\nu$ is invariant by rotations. Moreover, the summation is almost always finite and when it is there is only one term (see \cite{Bri2} Appendix). For this $i^{th}$ term we have 
$$\abs{h_\eps(t,x,v)} \leq e^{-\nu_0 t} \abs{f_0(x_i-t_iv_i,v_i)}\leq e^{-\nu_0 t} \norm{f_0}_{l^\infty_{x,v}},$$
which is the desired exponential decay by taking the weak-* limit of $h_\eps$ in $L^\infty_{t,x,v}$.
\end{proof}
\bigskip


\subsection{The case of Maxwellian diffusion}\label{subsec:GnuMD}

The diffusion operator on the boundary does not have a norm equals to one, the latter norm heavily depends on the weight of the space. The exponential decay is delicate since we do not have an explicit representation of $S_{G_\nu}(t)$ along characteristic trajectories. One needs to control the characteristic trajectories that do not reach the plane $\br{t=0}$ in time $t$. As we shall see, this number of problematic trajectories is small when the number of rebounds is large and so can be controlled for long times.

\bigskip
\begin{prop}\label{prop:semigroupGnuMD}
Let $q\in \br{1,\infty}$ $m=e^{\kappa\abs{v}^\alpha}$ with $\kappa >0$ and $\alpha$ in $(0,2)$ or $m=\langle v \rangle^k$ with $k > 2^{1/q}4^{1-1/q}$; let $f_0$ be in $L^q_vL^\infty_x\pa{m}$. Then there exists a unique solution $S_{G_\nu}(t)f_0 \in L^q_vL^\infty_x\pa{m}$ to 
\begin{equation}\label{eqGnuMD}
\cro{\partial_t + v\cdot\nabla_x + \nu(v)}\pa{S_{G_\nu}(t)f_0}=0
\end{equation}
such that $\restr{\pa{S_{G_\nu}(t)f_0}}{\Lambda} \in L^qL^\infty_\Lambda\pa{m}$ and satisfying the Maxwelian diffusion $\eqref{MD}$ with initial data $f_0$. Moreover it satisfies
$$\forall \nu_0' <\nu_0,\:\exists \: C_{\nu_0'}>0,\:\forall t\geq 0, \quad \norm{S_{G_\nu}(t)f_0}_{L^q_vL^\infty_x\pa{m}} \leq C_{\nu_0'}e^{-\nu'_0t}\norm{f_0}_{L^q_vL^\infty_x\pa{m}},$$
with $\nu_0 = \inf\br{\nu(v)}>0$.
\end{prop}
\bigskip

\begin{proof}[Proof of Proposition $\ref{prop:semigroupGnuMD}$]

We first prove uniqueness, then existence and finally exponential decay of solutions.

\bigskip
\textbf{Step 1: Uniqueness.} Assume that there exists such a solution $f$ in $L^\infty_{x,v}\pa{m}$.  The choice of weight implies
$$m(v)^{-1}\pa{1+\abs{v}} \in L^1_v,$$
and hence $f$ belongs to $L^1_{x,v}$ and $\restr{f}{\Lambda}$ belongs to $L^1_\Lambda$. We can therefore use the divergence theorem and the fact that $\nu(v)\geq \nu_0>0$:
\begin{eqnarray}
\frac{d}{dt}\norm{f}_{L^1_{x,v}} &=& \int_{\Omega\times\R^3}\mbox{sgn}(f(t,x,v))\cro{-v\cdot\nabla_x - \nu(v)}f(t,x,v)\:dxdv\nonumber
\\&=& -\int_{\Omega\times\R^3} v\cdot\nabla_x\pa{\abs{f}}\:dxdv - \norm{\nu(v)f}_{L^1_{x,v}}\nonumber
\\&\leq& -\int_{\Lambda}\abs{f(t,x,v)}\pa{v\cdot n(x)}\:dS(x)dv - \nu_0\norm{f}_{L^1_{x,v}}.\label{equniquenessL1}
\end{eqnarray}

\par Then using the change of variable $ v \mapsto \mathcal{R}_x(v)$, which has jacobian one, we have the boundary conditions $\eqref{MD}$
$$\int_{\Lambda^-}\abs{P_{\Lambda}(f)(x,v)}\abs{v\cdot n(x)}\:dS(x)dv \leq \int_{\Lambda^+}\abs{f(t,x,v_*)}\abs{v_*\cdot n(x)}\:dS(x)dv_*,$$
which implies that the integral on the boundary is positive. Hence uniqueness follows from a Gr\"onwall lemma.
\par The case $q=1$ is dealt with the same way since $L^1_vL^\infty_x\pa{m} \subset L^1_{x,v}$ and also $L^1L^\infty_\Lambda\pa{m} \subset L^1L^\infty_\Lambda$.

\bigskip
\textbf{Step 2: Existence.} Let $f(t,x,v) \in L^q_vL^\infty_x\pa{m}$ be a solution to $\eqref{eqGnuMD}$ satisfying Maxwellian diffusion boundary conditions and $\restr{f}{\Lambda} \in L^qL^\infty_{\Lambda}\pa{m}$. Then $h(t,x,v) = m(v)f(t,x,v)$ belongs to $L^q_vL^\infty_x$ with $\restr{h}{\Lambda} \in L^1L^\infty_{\Lambda}$. Moreover, $h$ satisfies the differential equation $\eqref{eqGnuMD}$ with the following boundary condition for all $t>0$
\begin{equation}\label{boundarynuMD}
\forall (x,v) \in \Lambda^-, \quad h(t,x,v) = c_\mu m(v)\mu(v)\int_{v_*\cdot n(x)>0} h(t,x,v_*)m(v_*)^{-1}\abs{v_*\cdot n(x)}\:dv_*.
\end{equation}
In order to work without weight we will prove the existence of $f\in L^q_vL^\infty_x$ such that $\restr{f}{\Lambda} \in L^qL^\infty_{\Lambda}$ and $f$ satisfies
$$\cro{\partial_t + v\cdot\nabla_x + \nu(v)}f=0$$
with the new diffusive condition $\eqref{boundarynuMD}$. And we will prove exponential decay in $L^q_vL^\infty_x$ for this function $f$.

\bigskip
To prove existence we consider the following iterative scheme with $h^{(l)} \in L^q_vL^\infty_x$ and $\restr{h^{(l)}}{\Lambda} \in L^qL^\infty_\Lambda$:
$$\cro{\partial_t + v\cdot\nabla_x +\nu} h^{(l)} =0,\quad h^{(l)}(0,x,v)= f_0(x,v)\mathbf{1}_{\br{\abs{v}\leq l}}$$
with the absorbed diffusion boundary condition for $t>0$ and $(x,v)$ in $\Lambda^-$
\begin{eqnarray}
h^{(l)}(t,x,v) &=& P^{(l)}_{\Lambda,m}(\restr{h^{(l)}}{\Lambda^+})(t,x,v) \nonumber
\\&=& \pa{1-\frac{1}{l}}c_\mu m(v)\mu(v)\int_{v_*\cdot n(x)>0} h^{(l)}(t,x,v_*)m(v_*)^{-1}\abs{v_*\cdot n(x)}\:dv_*.\label{Plkrho}
\end{eqnarray}
\par Again, multiplying $h^{(l)}$ by the appropriate weight raise the uniqueness of such a $h^{(l)}$ for any given $l$. The existence is proved  \textit{via} $\tilde{h}^{(l)} = m(v)^{-1}\mu^{-1}h^{(l)}$ since it satisfies $\cro{\partial_t -v\cdot\nabla_x -\nu(v)}\tilde{h}^{(l)} =0$ with the boundary condition
$$\tilde{h}^{(l)}(t,x,v) = \pa{1-\frac{1}{l}}\int_{v_*\cdot n(x)>0} \tilde{h}^{(l)}(t,x,v_*)c_\mu \mu(v_*)\abs{v_*\cdot n(x)}\:dv_*$$
and the initial data
$$\norm{\tilde{h}_0^{(l)}}_{L^q_vL^\infty_x} = \norm{m^{-1}\mu^{-1}f_0\mathbf{1}_{\br{\abs{v}\leq l}}}_{L^q_vL^\infty_x}\leq C_{l,m}\norm{f_0}_{L^q_vL^\infty_x}.$$
The boundary operator from $L^qL^\infty_{\Lambda^+}$ to $L^qL^\infty_{\Lambda^-}$ applied to $\tilde{h}^{(l)}$ is bounded by $(1-l^{-1}) <1$ and therefore $\tilde{h}^{(l)}\in L^q_vL^\infty_x$ exists with its restriction in $L^qL^\infty_\Lambda$(see \cite{BePro}). Thus the existence of $h^{(l)}$. The proof that $h^{(l)}$ is indeed in $L^q_vL^\infty_x$ and converges as $l$ tends to infinity will be done within the proof of exponential decay uniformly in $l$.

\bigskip
\textbf{Step 3: Exponential decay.}
As for the specular case $\eqref{representationGnueps}$, we can use the flow of characteristic to obtain a representation of $h^{(l+1)}$ in terms of $f_0$ and $h^{(l)}$. We recall the boundary operator $P^{(l)}_{\Lambda,m}$ $\eqref{Plkrho}$ and for all $(x,v) \notin \Lambda_0\cup\Lambda^-$,
\begin{equation}\label{representationGnuepsMD}
\begin{split}
h^{(l)}(t,x,v) =& \mathbf{1}_{t_1(x,v) \leq 0}\: e^{-\nu(v)t}f_0(x-tv,v)\mathbf{1}_{\br{\abs{v}\leq l}} 
\\&+ \mathbf{1}_{t_1(x,v) > 0}\: e^{-\nu(v)(t-t_1)}P^{(l)}_{\Lambda,m}\pa{\restr{h^{(l)}}{\Lambda^+}}(t_1,x_1,v),
\end{split}
\end{equation}
where we defined $t_1=t - t_{min}(x,v)$ and $x_1(x,v) = x-(t-t_1(x,v))v$.
\par The idea is to iterate the latter representation inside the integral term $P^{(l)}_{\Lambda,m}$. This leads to a sequence of functions $(t_p,x_p,v_p)$ depending on the independent variables $(t_i,x_i,v_i)_{0\leq i \leq p-1}$ with $(t_0,x_0,v_0) = (t,x,v)$.

\bigskip
To shorten notations we define the probability measure on $\Lambda^+$
$$d\sigma_x(v) = c_\mu \mu(v) \abs{v\cdot n(x)}\:dv$$
and remark that the boundary condition $\eqref{Plkrho}$ becomes
$$h^{(l)}(t,x,v) = \pa{1-\frac{1}{l}}\frac{1}{\tilde{m}(v)}\int_{v_*\cdot n(x)>0} h^{(l)}(t,x,v_*)\tilde{m}(v_*)\:d\sigma_x(v_*)$$
with
\begin{equation}\label{tildemkrho}
\tilde{m}(v) = \frac{1}{c_\mu \mu(v) m(v)}.
\end{equation}
With these notations one can derive the following implicit iterative representation of $h^{(l)}$.  We refer to \cite{Gu6} Lemma $24$ and $(208)$ for a rigorous induction.
\begin{itemize}
\item If $t_1\leq 0$ then 
\end{itemize}
\begin{equation}\label{startinductionMD}
h^{(l)}(t,x,v) = e^{-\nu(v)t}f_0(x-tv,v)\mathbf{1}_{\br{\abs{v}\leq l}};
\end{equation}
\begin{itemize}
\item If $t_1 > 0$ then for all $p\geq 2$, 
\end{itemize}
\begin{equation}\label{inductionMD}
\begin{split}
&h^{(l)}(t,x,v) 
\\&= \frac{1}{\tilde{m}(v)}e^{-\nu(v)(t-t_1)}\sum\limits_{i=1}^{p}\pa{1-\frac{1}{l}}^i\int_{\prod\limits_{j=1}^{p}\br{v_j \cdot n(x_i)>0}}\mathbf{1}_{[t_{i+1},t_i)}(0)\:h^{(l)}_0(x_i - t_iv_i,v_i)d\Sigma_i(0)
\\&\quad +\pa{1-\frac{1}{l}}^p\frac{1}{\tilde{m}(v)}e^{-\nu(v)(t-t_1)}\int_{\prod\limits_{j=1}^{p}\br{v_j \cdot n(x_i)>0}}\mathbf{1}_{t_{p+1}>0}\:h^{(l)}(t_p,x_p,v_p)d\Sigma_{p}(t_p),
\end{split}
\end{equation}
where
$$d\Sigma_{i}(s) = e^{-\nu(v_i)(t_i-s)}\tilde{m}(v_i)\pa{\prod\limits_{j=1}^{i-1}e^{-\nu(v_j)(t_j-t_{j+1})}} \:d\sigma_{x_1}(v_1)\dots d\sigma_{x_{p}}(v_{p}).$$
The last term on the right-hand side of $\eqref{inductionMD}$ represents all the possible trajectories that are still able to generate new trajectories after $p$ rebounds. The first term describes all the possible trajectories reaching the initial plane $\br{t=0}$ in at most $p$ rebounds.

\bigskip
Computations are similar either $q=1$ or $q=\infty$. For $q=\infty$, it is enough to bound $\pa{h^{(l)}}_{l\in\N}$ to obtain weak-* convergence whereas $q=1$ requires more efforts. We therefore only deal with $q=1$ and point out the few differences for $q=\infty$ in Remark $\ref{rem:differenceL1Linftynu}$.
\par We will prove that $\norm{h^{(l)}}_{L^1_vL^\infty_x}$ satisfies an exponential decay uniformly in $l$ and then show that $\pa{h^{(l)}}_{l\in\N}$ (resp. its restrictions on $\Lambda^+$ and $\Lambda^-$) is weakly compact in $L^\infty_tL^1_vL^\infty_x$ (resp. on $\Lambda^+$ and $\Lambda^-$). The proof will be done in three steps. We first study the sequence $\pa{h^{(l)}\mathbf{1}_{t_1\leq 0}}_{l\in\N}$ in $L^\infty_tL^1_vL^\infty_x$, then $\pa{h^{(l)}\mathbf{1}_{t_1>0}}_{l\in\N}$ in $L^\infty_{[0,T_0]}L^1_vL^\infty_x$ with $T_0$ large and finally $\pa{h^{(l)}\mathbf{1}_{t_1>0}}_{l\in\N}$.

\bigskip
\textbf{Step 3.1: $\mathbf{\br{t_1\leq 0}}$.}
We first use $\eqref{startinductionMD}$ for all $l$ in $\N$ and all $t\geq 0$,
\begin{equation}\label{expodecayt1leq0}
\norm{h^{(l)}\mathbf{1}_{t_1\leq 0}(t,\cdot,\cdot)}_{L^1_vL^\infty_x} \leq  e^{-\nu_0 t}\norm{f_0}_{L^1_vL^\infty_x}.
\end{equation}
And also for all measurable set $K \subset \R^3$,
\begin{equation}\label{dunfordpettist1leq0}
\int_{K}\sup\limits_{x \in \Omega}\abs{h^{(l)}(t,x,v)\mathbf{1}_{t_1\leq 0}} \:dv \leq \int_{K}\sup\limits_{x \in \Omega}\abs{f_0(t,x,v)} \:dv.
\end{equation}
$f_0$ belongs to $L^1_vL^\infty_x$ and therefore the latter inequality implies that the sequence $\pa{\sup\limits_{x\in \Omega}\abs{h^{(l)}\mathbf{1}_{t_1\leq 0}(t,x,\cdot)}}_{l\in\N}$ is bounded and equi-integrable. The latter is also true restricted to $\Lambda^+$ since in that case $\restr{h^{(l)}\mathbf{1}_{t_1\leq 0}}{\Lambda^+} =0$.

\bigskip
\textbf{Step 3.2: $\mathbf{\br{t_1 > 0}}$ and $\mathbf{0\leq t \leq T_0}$.}
We focus on the case $t_1 >0$.
\par The exponential decay in $d\Sigma_i(s)$ is bounded by $e^{-\nu_0(t_1-s)}$ and we notice that the definition of $\tilde{m}$ $\eqref{tildemkrho}$ implies
$$\tilde{m}(v)d\sigma_x(v) = m(v)^{-1}\abs{v\cdot n(x)}dv$$
We first take the supremum over $x$ in $\Omega$ and then integrate in $v$ over $\R^d$ the first term on the right-hand side of $\eqref{inductionMD}$ and we obtain the following upper bound
\begin{equation}\label{startrighthandside}
\begin{split}
&e^{-\nu_0 t} \int_{\R^3}\frac{dv}{\tilde{m}(v)}\Big\{
\\&\sup\limits_{x\in\Omega}\sum\limits_{i=1}^p\int_{\prod\limits_{\overset{j=1}{j\neq i}}^{p}\br{v_j \cdot n(x_i)>0}}\mathbf{1}_{[t_{i+1},t_i)}(0)\pa{\int_{\R^3}\sup\limits_{y\in\Omega}\abs{h^{(l)}_0(y,v_i)}\frac{\abs{v_i}}{m(v_i)}\:dv_i}\:\prod\limits_{\overset{j=1}{j\neq i}}^{p}d\sigma_{x_j}(v_j)\Big\}.
\end{split}
\end{equation}
Since there exists $C_m >0$ (note that in what follows $C_m$ will stand for any explicit positive constant only depending on $m$) such that 
$$\frac{\abs{v}}{m(v)}\leq C_m$$
we can further bound $\eqref{startrighthandside}$ by
\begin{equation}\label{righthansidefinal}
\begin{split}
&C_m e^{-\nu_0 t} \pa{\int_{\R^3}\frac{dv}{\tilde{m}(v)}}\norm{f_0}_{L^1_vL^\infty_x}\sup\limits_{x,v}\sum\limits_{i=1}^p\int_{\prod\limits_{j=1}^{p}\br{v_j \cdot n(x_i)>0}}\mathbf{1}_{[t_{i+1},t_i)}(0)\:d\sigma_{x_1}\dots d\sigma_{x_p}
\\&\quad\quad\quad\leq C_m e^{-\nu_0 t} \pa{\int_{\R^3}\frac{dv}{\tilde{m}(v)}}\norm{f_0}_{L^1_vL^\infty_x}
\\ &\quad\quad\quad\leq C_m e^{-\nu_0 t}\norm{f_0}_{L^1_vL^\infty_x},
\end{split}
\end{equation}
where we used the fact that $\int_{v_i\cdot n(x_i)>0}d\sigma_{x_i}(v_i)=1$ and the following control
\begin{equation}\label{integralmtilde}
\int_{\R^3}\frac{dv}{\tilde{m}(v)} \leq C_m.
\end{equation}

\bigskip
We now turn to the study of the second term on the right-hand side of $\eqref{inductionMD}$.
\par We first notice that on the set $\br{t_{p+1}>0}$ we have $t_1(t_p,x_p,v_p)>0$ and therefore
\begin{equation}\label{tp+1tot1}
\mathbf{1}_{t_{p+1}>0}\abs{h^{(l)}(t_p,x_p,v_p)}\leq \mathbf{1}_{t_{p}>0}\sup\limits_{y\in\Omega}\abs{h^{(l)}(t_p,y,v_p)\mathbf{1}_{t_{1}>0}}.
\end{equation}

\par We take the supremum in $x\in\Omega$ and integrating in $v$ over $\R^3$ the second term on the right-hand side of $\eqref{inductionMD}$ and make the same computations as for the first term. This yields the following upper bound for $0\leq t \leq T_0$
\begin{equation*}
\begin{split}
&C_m e^{-\nu_0 (t-t_1)} \pa{\int_{\R^3}\frac{dv}{\tilde{m}(v)}\sup\limits_{x\in\Omega}\int_{\prod\limits_{j=1}^{p}\br{v_j \cdot n(x_i)>0}}e^{-\nu_0(t_1-t_p)}\sup\limits_{y\in\Omega}\abs{h^{(l)}(t_p,y,v_p)\mathbf{1}_{t_{1}>0}}\mathbf{1}_{t_p>0}}
\\&\leq C_m e^{-\nu_0 t} \pa{\int_{\R^3}\frac{dv}{\tilde{m}(v)}}\sup\limits_{0\leq s \leq T_0}\pa{\cro{e^{\nu_0s}\norm{h^{(l)}\mathbf{1}_{t_1>0}}}\sup\limits_{x,v}\int_{\prod\limits_{j=1}^{p}\br{v_j \cdot n(x_i)>0}}\mathbf{1}_{t_p>0}\prod\limits_{j=1}^p d\sigma_{x_i}}
\end{split}
\end{equation*}

\par As said at the beginning of the section, the trajectories not hitting the initial plane after $p$ rebounds is small when $p$ becomes large. This is given by \cite{EGKM} Lemma $4.1$ which states that there exist $C_1$, $C_2 >0$ such that for all $T_0$ sufficiently large, taking $p=C_1 T_0^{5/4}$ yields
\begin{equation}\label{controlreboundT0}
\forall 0\leq s\leq T_0,\: \forall x \in \bar{\Omega},\:\forall v \in\R^3, \quad \int_{\prod\limits_{j=1}^{p}\br{v_j \cdot n(x_i)>0}}\mathbf{1}_{t_p>0}\prod\limits_{j=1}^p d\sigma_{x_i} \leq \pa{\frac{1}{2}}^{C_2T_0^{5/4}}.
\end{equation}
Plugging it into the last inequality yields the following bound for the second term on the right-hand side of $\eqref{inductionMD}$ for all $t$ in $[0,T_0]$
\begin{equation}\label{finalrighthandside}
\begin{split}
&C_m e^{-\nu_0 t}\pa{\frac{1}{2}}^{C_2T_0^{5/4}}\pa{\int_{\R^3}\frac{dv}{\tilde{m}(v)}}\sup\limits_{0\leq s \leq T_0}\cro{e^{\nu_0s}\norm{h^{(l)}\mathbf{1}_{t_1>0}}_{L^1_vL^\infty_x}}
\\&\quad\quad\quad\quad\quad\leq C_m e^{-\nu_0 t}\pa{\frac{1}{2}}^{C_2T_0^{5/4}}\sup\limits_{0\leq s \leq T_0}\cro{e^{\nu_0s}\norm{h^{(l)}\mathbf{1}_{t_1>0}}_{L^1_vL^\infty_x}},
\end{split}
\end{equation}
where we used $\eqref{integralmtilde}$.

\bigskip
Gathering $\eqref{righthansidefinal}$ and $\eqref{finalrighthandside}$ gives
\begin{equation*}
\begin{split}
\sup\limits_{0\leq t \leq T_0}\cro{e^{\nu_0 t}\norm{h^{(l)}\mathbf{1}_{t_1>0}}_{L^1_vL^\infty_x}} \leq& C_m\norm{f_0}_{L^1_vL^\infty_x}
\\& + C_m\pa{\frac{1}{2}}^{C_2T_0^{5/4}}\sup\limits_{0\leq t \leq T_0}\cro{e^{\nu_0 t}\norm{h^{(l)}\mathbf{1}_{t_1>0}}_{L^1_vL^\infty_x}}.
\end{split}
\end{equation*}
Choosing $T_0$ even larger if need be such that 
$$C_m\pa{\frac{1}{2}}^{C_2T_0^{5/4}} \leq \frac{1}{2},$$
gives
\begin{equation}\label{expodecayt1>0<rho}
\exists C_{m} >0,\:\forall t \in [0,T_0], \quad \norm{h^{(l)}\mathbf{1}_{t_1>0}(t,\cdot,\cdot)}_{L^1_vL^\infty_x} \leq C_{m} e^{-\nu_0 t} \norm{f_0}_{L^1_vL^\infty_x}.
\end{equation}
Moreover, in $\eqref{righthansidefinal}$ and $\eqref{finalrighthandside}$ we kept the dependencies in the integration against $v$ in $\R^3$. Taking the integration over a measurable set $K\subset\R^3$, the same computations and the same choice of $T_0$ would give
\begin{equation}\label{dunfordpettist1>0<rho}
\exists C_{m} >0,\:\forall t \in [0,T_0], \quad \int_K \sup\limits_{x\in\Omega}\abs{h^{(l)}\mathbf{1}_{t_1>0}(t,x,v)}\:dv \leq C_{m}\norm{f_0}_{L^1_vL^\infty_x} \int_{K}\frac{dv}{\tilde{m}(v)}.
\end{equation}

\par Since $\tilde{m}^{-1}$ is integrable on $\R^3$, $\pa{\sup\limits_{x\in \Omega}\abs{h^{(l)}\mathbf{1}_{t_1> 0}(t,x,\cdot)}}_{l\in\N}$ is bounded and equi-integrable.

\bigskip
\textbf{Step 3.3: conclusion}
The constant $C_{m}$ in $\eqref{expodecayt1>0<rho}$ does not depend on $T_0$. Therefore, for any $\nu_0' < \nu_0$ we can choose $T_0 = T_0(m,\nu_0')$ large enough so that $\eqref{expodecayt1>0<rho}$ holds for $0\leq t \leq T_0$ and $C_{m}e^{-\nu_0 T_0} \leq e^{-\nu_0'T_0}$.
\par For that specific $T_0$ one has
$$\norm{h^{(l)}(T_0)}_{L^1_vL^\infty_x} \leq e^{-\nu_0'T_0} \norm{f_0}_{L^1_vL^\infty_x}.$$

\bigskip
We could now start the proof at $T_0$ up to $2T_0$ and iterating this process we get
\begin{equation*}
\begin{split}
\forall n\in\N,\quad \norm{h^{(l)}(nT_0)\mathbf{1}_{t_1>0}}_{L^1_vL^\infty_x} &\leq e^{-\nu_0'T_0} \norm{h^{(l)}((n-1)T_0}_{L^1_vL^\infty_x}
\\&\leq e^{-2\nu_0'T_0} \norm{h^{(l)}((n-2)T_0)}_{L^1_vL^\infty_x}
\\&\leq \dots \leq e^{-\nu_0'nT_0} \norm{f_0}_{L^1_vL^\infty_x}.
\end{split}
\end{equation*}
Finally, for all $t$ in $[nT_0,(n+1)T_0]$ we apply $\eqref{expodecayt1>0<rho}$ with the above to get
\begin{equation*}
\begin{split}
\norm{h^{(l)}\mathbf{1}_{t_1>0}(t,\cdot,\cdot)}_{L^1_vL^\infty_x} &\leq C_m e^{-\nu_0 (t-nT_0)} \norm{h^{(l)}(nT_0)}_{L^1_vL^\infty_x}
\\&\leq C_me^{-\nu_0 t+ \pa{\nu_0-\nu_0'}nT_0}\norm{f_0}_{L^1_vL^\infty_x}.
\end{split}
\end{equation*}
Hence the uniform control in $t$,
\begin{equation}\label{expodecaytotal}
\exists C_m >0,\:\forall t\geq 0, \quad \norm{h^{(l)}\mathbf{1}_{t_1>0}(t,\cdot,\cdot)}_{L^1_vL^\infty_x} \leq C_m e^{-\nu'_0 t} \norm{f_0}_{L^1_vL^\infty_x}.
\end{equation}
Again, we could only integrate over a measurable set $K\subset\R^3$, the same computations and the same choice of $T_0$ would give
\begin{equation}\label{dunfordpettistotal}
\exists C_k >0,\:\forall t \geq 0, \quad \int_K \sup\limits_{x\in\Omega}\abs{h^{(l)}\mathbf{1}_{t_1>0}(t,x,v)}\:dv \leq C_m\norm{f_0}_{L^1_vL^\infty_x} \int_{K}\frac{dv}{\tilde{m}(v)}.
\end{equation}

\bigskip
Combining $\eqref{expodecayt1leq0}$-$\eqref{expodecaytotal}$ we see that $\pa{h^{(l)}}_{l\in\N}$ is bounded in $L^\infty_tL^1_vL^\infty_v$. Moreover, by $\eqref{dunfordpettist1leq0}$-$\eqref{dunfordpettistotal}$ $\pa{\sup\limits_{x\in\Omega}\abs{h^{(l)}}}_{l\in\N}$ is equi-integrable on $L^1_v$. We can therefore apply the Dunford-Pettis theorem for $L^1$ combined with weak-compactness property of $L^\infty$ and find that $\pa{h^{(l)}}_{l\in\N}$ converges (up to a subsequence) weakly-* in $L^\infty_tL^1_vL^\infty_x$. The limit $f$ in $L^\infty_tL^1_vL^\infty_x$ is solution to the linear equation $\partial_t f = G_\nu f$ with initial data $f_0$.
\par Besides, since we always bound the integral of $h^{(l)}$ on $\br{v_i\cdot n(x_i)}>0$ by its integral on $\R^3$, we could do the same computations for $\restr{h^{(l)}}{\Lambda^+}$ by keeping the integral on $\br{v_i\cdot n(x_i)}>0$. Dunford-Pettis theorem again  and then the boundary conditions implies that $\pa{\restr{h^{(l)}}{\Lambda}}_{l\in\N}$ converges weakly-* in $L^\infty_tL^1L^\infty_\Lambda$.
\par $f$ thus satisfies the diffusion boundary condition $\eqref{boundarynuMD}$, has its restriction on $\Lambda$ in $L^1L^\infty_\Lambda$ and the exponential decay holds (since it holds for all $h^{(l)}$ uniformly in $l$). Which concludes the proof of existence and exponential decay.

\bigskip
\begin{remark}\label{rem:differenceL1Linftynu}
The case $q=\infty$ is dealt with the same way since the function that we bound $\abs{v}m(v)^{-1}$ and the one we integrate $\tilde{m}(v)^{-1}$ are respectively integrable and bounded since $k>4$ in the case $q=\infty$. Therefore in $\eqref{startrighthandside}$ we can take out the $L^\infty_{x,v}$-norm of $h_0^{(l)}$ and have $C_m$ be the integral of $\abs{v}m(v)^{-1}$ and finally bound $\tilde{m}(v)^{-1}$ instead of integrating it in $\eqref{righthansidefinal}$. This leads to the same estimates in $L^\infty_{x,v}$.
\end{remark}
\end{proof}
\bigskip

The proof above can be adapted to obtain that $S_{G_\nu}(t)$ controls `a bit more' than the mere $L^q_vL^\infty_{x}\pa{m}$-norm. This property will play a key role in the nonlinear case.
\par In $\eqref{startrighthandside}$ one could multiply and divide by $\nu(v_i)$ and the function
$$\frac{\abs{v_i}\nu(v_i)}{m(v_i)}$$
is still bounded (resp. integrable) on $\R^3$ if $m$ is a stretch exponential or if $m = \langle v \rangle^k$ with $k>1+\gamma$ (resp. $k>4+\gamma$). The conclusion $\eqref{expodecayt1>0<rho}$ then becomes
$$\exists C_m >0,\:\forall t \in [0,T_0], \quad \norm{h^{(l)}\mathbf{1}_{t_1>0}(t,\cdot,\cdot)}_{L^q_vL^\infty_{x}} \leq C_m e^{-\nu_0 t} \norm{f_0}_{L^q_vL^\infty_{x}\pa{\nu^{-1}}}.$$
This inequality is true at $T_0$ so using the induction that lead to $\eqref{expodecaytotal}$ and using the latter gives the following corollary.

\bigskip
\begin{cor}\label{cor:SGnuMDcontrolnu-1}
Let $q\in \br{1,\infty}$ $m=e^{\kappa\abs{v}^\alpha}$ with $\kappa >0$ and $\alpha$ in $(0,2)$ or $m=\langle v \rangle^k$ with $k > 2^{1/q}4^{1-1/q}+\gamma$; let $f_0$ be in $L^q_vL^\infty_{x}\pa{m}$. Then the solution $S_{G_\nu}(t)f_0 \in L^q_vL^\infty_{x}\pa{m}$ built in Proposition $\ref{prop:semigroupGnuMD}$ satisfies for all $\nu_0'<\nu_0$
$$\exists \: C_{\nu_0'}>0,\:\forall t\geq 0, \quad \norm{S_{G_\nu}(t)\pa{f_0}\mathbf{1}_{t_1>0}(t,\cdot,\cdot)}_{L^q_vL^\infty_{x}\pa{m}} \leq C_m e^{-\nu'_0 t} \norm{f_0}_{L^q_vL^\infty_{x}\pa{m\nu^{-1}}}.$$
\end{cor}
\bigskip

%% file: L2Linftytheory.tex
\section{Review of the $L^2-L^\infty$ theory for the full linear part}\label{sec:L2Linftytheory}

As discussed in the introduction, a mixed $L^2-L^\infty$ theory has been developed \cite{Gu6}\cite{EGKM} for the linear perturbed Boltzmann equation
\begin{equation}\label{lineareq}
\partial_t f + v\cdot\nabla_x f = L(f),
\end{equation}
together with boundary conditions. The idea of studying the possible generation of a semigroup with exponential decay in $L^2_{x,v}\pa{\mu^{-1/2}}$ by
$$G = -v\cdot\nabla_x + L,$$
together with boundary conditions,  is a natural one because of Subsection \ref{subsec:background}. 

\bigskip
This section is devoted to the description of the $L^2-L^\infty$ theory developed first by Guo \cite{Gu6} and extended by Esposito, Guo, Kim and Marra \cite{EGKM}. This theory will be the starting point of our main proofs.
\bigskip


\subsection{$L^2_{x,v}(\mu^{-1/2})$ theory for the linear perturbed operator}\label{subsec:L2theory}

As seen discussed before, the general theory \cite{BePro} for equations of the form 
$$\partial_t f + v\cdot\nabla_x f = g$$
fails for specular reflections or Maxwellian diffusion boundary conditions because the boundary operator $P$ is of norm one. However, restricting the Maxwellian diffusion to the set of functions in $L^2_{x,v}$ satisfying the preservation of mass implies that, in some sense, $\norm{P} <1$ (mere strict Cauchy-Schwarz inequality). One can therefore hope to develop a semigroup theory for $G$ in $L^2_{x,v}\pa{\mu^{-1/2}}$ with mass conservation. This has been recently achieved by constructive methods \cite{EGKM}. They proved the following theorem (see \cite{EGKM} Theorem $6.1$ with $g=r=0$).

\bigskip
\begin{theorem}\label{theo:semigroupL2MD}
Let $f_0$ be in $\mbox{MD}\cro{L^2_{x,v}\pa{\mu^{-1/2}}}$. Then there exists a unique mass preserving solution $S_G(t)f_0 \in L^2_{x,v}\pa{\mu^{-1/2}}$ to the linear perturbed Boltzmann equation $\eqref{lineareq}$ with Maxwellian diffusion boundary condition $\eqref{MD}$.
\\ Moreover there exist explicit $C_G$, $\lambda_G >0$, independent of $f_0$, such that
$$\forall t \geq 0, \quad \norm{S_G(t)f_0}_{L^2_{x,v}\pa{\mu^{-1/2}}} \leq C_G e^{-\lambda_G t}\norm{f_0}_{L^2_{x,v}\pa{\mu^{-1/2}}}.$$
\end{theorem}
\bigskip

\bigskip
Unfortunately, in the case of specular reflections the uniqueness is not true in general due to a possible blow-up of the $L^2_{\mbox{\scriptsize{loc}}}\pa{\R^+;L^2_{x,v}\pa{\Lambda}}$ at the grazing set $\Lambda_0$ \cite{Uk3,BePro,CIP}. However, an \textit{a priori} exponential decay of solutions is enough to obtain an $L^\infty$ theory provided that we endow the space with a strong weight (see next subsection). Such an \textit{a priori} study has been derived in \cite{Gu6} by a contradiction argument.

\bigskip
\begin{theorem}\label{theo:semigroupL2SR}
Let $f_0$ be in $\mbox{SR}\cro{L^2_{x,v}\pa{\mu^{-1/2}}}$. Suppose that $f(t,x,v)$ is a solution to the linear perturbed Boltzmann equation $\eqref{lineareq}$ in $\mbox{SR}\cro{L^2_{x,v}\pa{\mu^{-1/2}}}$ with initial data $f_0$ and satisfying the specular reflections boundary condition $\eqref{SR}$. Suppose also that $\restr{f}{\Lambda}$ belongs to $L^2_\Lambda\pa{\mu^{-1/2}}$. Then there exists $C_G$, $\lambda_G>0$ such that
$$\forall t \geq 0, \quad \norm{f(t)}_{L^2_{x,v}\pa{\mu^{-1/2}}} \leq C_G e^{-\lambda_G t}\norm{f_0}_{L^2_{x,v}\pa{\mu^{-1/2}}}.$$
The constants $C_G$ and $\lambda_G$ are independent of $f$.
\end{theorem}
\bigskip

Note that the two previous theorems hold for $\Omega$ being a $C^1$ bounded domain.
\bigskip


\subsection{The $L^\infty$ framework}\label{subsec:Linftytheory}

It has been proved in \cite{Gu6} section $4$ that if Theorem \ref{theo:semigroupL2MD} and Theorem \ref{theo:semigroupL2SR} hold true then one can develop and $L^\infty_{x,v}\pa{\langle v \rangle^\beta\mu^{-1/2}}$ theory for the semigroup generated by $G=-v\cdot\nabla_x+L$, as long as $\beta$ is sufficiently large. We already discussed the fact that we do not have a semigroup property for $G$ in $L^2_{x,v}$ due to the possible lack of uniqueness. To overcome this inconvenient it is compulsory to go into $L^\infty_{x,v}\pa{\langle v \rangle^{\beta}\mu^{-1/2})}$ where trace theorems are known to hold and $\beta$ is large enough so that one can use the \textit{a priori} estimate in $L^2_{x,v}\pa{\mu^{-1/2}}$ thanks to a change of variable along the characteristic trajectories (which requires the strict convexity and the analyticity of the domain in the case of specular reflexions).
\par In other words, having only an $\textit{a priori}$ exponential decay in $L^2_{x,v}$ can be used to obtain that $G$ actually generates an exponentially decaying semigroup in $L^\infty$, where we have Ukai's trace theorem \cite{Uk3} Theorem $5.1.1$ to have well-defined restrictions at the boundary and therefore uniqueness of solutions.
\par Moreover, this semigroup theory is compatible with the remainder term $Q(f,f)$ and offers existence, uniqueness and solutions to the perturbed Boltzmann equation
\begin{equation}\label{perturbedeqfinal}
\partial_t f = G(f) + Q(f,f)
\end{equation}
in $L^\infty\pa{\langle v \rangle^\beta\mu^{-1/2}}$ as long as $\norm{f_0}_{L^\infty\pa{\langle v \rangle^\beta\mu^{-1/2}}}$ is small enough.
\par We state here a theorem adapted from \cite{Gu6}. The case of specular reflections is derived from Theorem $8$ and the proof of Theorem $3$ and the case of Maxwellian diffusion from Theorem $9$ and the proof of Theorem $4$. One can also look at \cite{EGKM} Theorem $1.3$ for a constructive proof in the case of Maxwellian diffusion boundary conditions.

\bigskip
\begin{theorem}\label{theo:semigroupLinfty}
Let $\Omega$ be a $C^1$ bounded domain if boundary conditions are Maxwellian diffusion and let $\Omega$ be analytic and strictly convex in the sense of $\eqref{strictlyconvex}$ if they are specular reflections. Define $w_\beta(v) = \langle v \rangle^\beta\mu(v)^{-1/2}$.
\\Then for all $\beta$ such that $\beta^{-2}\pa{1+\abs{v}}^3 \in L^1_v$ the operator $G=-v\cdot\nabla_x +L$ generates a semigroup $S_G(t)$ in $\mbox{SR}\cro{L^\infty_{x,v}(w_\beta)}$ and in $\mbox{MD}\cro{L^\infty_{x,v}(w_\beta)}$. Moreover, there exists $C_G$, $\lambda_G>0$ such that for all $f_0$ in $L^\infty_tL^\infty_{x,v}(w_\beta)$ satisfying the appropriate conservation laws and all $t \geq 0$
$$\norm{S_G(t)f_0}_{L^\infty_{x,v}(w_\beta)} \leq C_G e^{-\lambda_G t}\norm{f_0}_{L^\infty_{x,v}(w_\beta)},$$
and for all $0<\lambda_G' < \lambda_G$,
\begin{equation*}
\begin{split}
&\norm{\int_0^t S_G(t-s)Q(f_0,f_0)\:ds}_{L^\infty_{x,v}(w_\beta)} \leq C_Ge^{-\lambda_G' t}\sup\limits_{s\in[0,t]}\cro{e^{\lambda_G' s}\norm{f_0(s)}^2_{L^\infty_{x,v}(w_\beta)}}.
 \end{split}
 \end{equation*}
\end{theorem}
\bigskip

\begin{remark}\label{rem:analytic}
We emphasize that the strict convexity required in our Theorem $\ref{theo:cauchySR}$ only comes from the fact that such a geometric property is needed in order to apply the theorem above and thus having a well-established semigroup theory in the framework of specular reflections.
\par Moreover, if the proof in the case of Maxwellian diffusion has been made constructive \cite{EGKM}, the case of specular reflections heavily relies on a contradiction argument combined with analyticity (\cite{Gu6} Lemma $22$) and a constructive proof is still an open problem.
\end{remark}
\bigskip

%% file: extensionstretchpoly.tex
\section{System of equations solving the perturbed Boltzmann equation}\label{sec:extensionSRMD}

This section is dedicated to the proofs of Theorem $\ref{theo:cauchySR}$ and Theorem $\ref{theo:cauchyMD}$. The latter proofs rely on a specific decomposition of the operator $G=-v\cdot\nabla_x+L$ that allows to solve a system of differential equations that connect the larger spaces $L^\infty_{x,v}\pa{m}$ to the more regular space $L^\infty_{x,v}\pa{\langle v \rangle^\beta\mu^{-1/2}}$ where solutions to the perturbed Boltzmann equation are known to exists (see Subsection $\ref{subsec:Linftytheory}$). As said in the introduction. our method follows the recent extension methods for strongly continuous semigroups \cite{GMM} but an analytic adaptation has to be developed since we already saw that $S_{G}(t)$ is not necessarily strongly continuous.
\par Firstly, Subsection $\ref{subsec:decomposition}$ describes the strategy we shall use and presents the new system of equations we will solve. Then Subsection $\ref{subsec:f1}$ and Subsection $\ref{subsec:f2}$ solve the system of differential equations.
\bigskip


\subsection{Decomposition of the perturbed Boltzmann equation and toolbox}\label{subsec:decomposition}

The main strategy is to find a decomposition of the perturbed Boltzmann equation
$\eqref{perturbedBE}$ into a system of differential equations where we could use of the theory developed in $L^\infty_{x,v}$. More precisely, one would like to solve a somewhat simpler equation in $L^\infty_{x,v}\pa{m}$ and that the remainder part has regularising properties and thus be handled in the smaller space $L^\infty_{x,v}\pa{\langle v \rangle^\beta\mu^{-1/2}}$. Then the exponential decay in the more regular space could be carried up to the bigger space. One can easily see that for the weights considered in the present work and for $q=1$ or $q=\infty$ we have 
$$L^\infty_{x,v}\pa{\langle v \rangle^\beta\mu^{-1/2}} \subset L^q_vL^\infty_x\pa{m}.$$

\bigskip
We follow the decomposition of $G$ proposed in \cite{GMM}.
\par For $\delta$ in $(0,1)$, to be chosen later, we consider $\Theta_\delta = \Theta_\delta(v,v_*,\sigma)$ in $C^\infty$ that is bounded by one everywhere, is exactly one on the set
$$\left\{\abs{v}\leq \delta^{-1}    \quad\mbox{and}\quad 2\delta\leq\abs{v-v_*}\leq \delta^{-1}    \quad\mbox{and}\quad \abs{\mbox{cos}\:\theta} \leq 1-2\delta \right\}$$
and whose support is included in
$$\left\{\abs{v}\leq 2\delta^{-1}    \quad\mbox{and}\quad \delta\leq\abs{v-v_*}\leq 2\delta^{-1}    \quad\mbox{and}\quad \abs{\mbox{cos}\:\theta} \leq 1-\delta \right\}.$$
We define the splitting
$$G = A^{(\delta)} +B^{(\delta)},$$
with
$$A^{(\delta)} h (v) = C_\Phi\int_{\R^3\times\mathbb{S}^2}\Theta_\delta\left[\mu'_*h' + \mu'h'_* - \mu h_*\right]b\left(\mbox{cos}\:\theta\right)\abs{v-v_*}^\gamma\:d\sigma dv_*$$
and
$$B^{(\delta)} h (v) = B^{(\delta)}_{2} h (v) -\nu(v) h(v) -v\cdot\nabla_x h(v)= G_\nu h(v) + B^{(\delta)}_{2} h (v),$$
where
$$B^{(\delta)}_{2} h (v) = \int_{\R^3\times\mathbb{S}^2}\left(1-\Theta_\delta\right)\left[\mu'_*h' + \mu'h'_* - \mu h_*\right]b\left(\mbox{cos}\:\theta\right)\abs{v-v_*}^\gamma\:d\sigma dv_*.$$

\bigskip
$A^{(\delta)}$ is a kernel operator with a compactly supported kernel. Therefore it has the following regularising effect.

\bigskip
\begin{lemma}\label{lem:controlA}
For any $q$ in $\br{1,\infty}$, the operator $A^{(\delta)}$ maps $L^q_v$ into $L^q_v$ with compact support. More precisely, for all $\beta\geq 0$ and all $\alpha\geq 0$, there exists $R_{\delta}$ and $C_A = C\pa{\delta,q,\beta,\alpha}>0$  such that 
$$\forall h \in L^q_v,\: \emph{\mbox{supp}}\left(A^{(\delta)} h \right) \subset B(0,R_{\delta}), \quad \norm{A^{(\delta)} h}_{L^q_v\pa{\langle v\rangle^\beta\mu^{-\alpha}}} \leq C_A\norm{h}_{L^q_v}.$$
\end{lemma}
\bigskip

\begin{proof}[Proof of Lemma $\ref{lem:controlA}$]
The kernel of the operator $A^{(\delta)}$ is compactly supported so its Carleman representation (see \cite{Ca2} or \cite{Vi2}) gives the existence of $k^{(\delta)}$ in $C^\infty_c\left(\R^3\times\R^3\right)$ such that 
\begin{equation}\label{Adeltaepskernel} 
 A^{(\delta)} h(v) = \int_{\R^3} k^{(\delta)}(v,v_*)h(v_*)\:dv_*,
 \end{equation}
and therefore the control on $\norm{A^{(\delta)} h}_{L^q_v\pa{\langle v\rangle^\beta\mu^{-\alpha}}}$ is straightforward.
\end{proof}
\bigskip

Thanks to this regularising property of the operator $A^{(\delta)}$ we are looking for solutions to the perturbed Boltzmann equation
$$\partial_t f = Gf + Q(f,f)$$
in the form of $f=f_1+f_2$ with $f_1$ in $L^\infty_{x,v}\pa{m}$ and $f_2$ in $L^\infty_{x,v}\pa{\langle v \rangle^\beta\mu^{-1/2}}$ and $(f_1,f_2)$ satisfying the following system of equation

\begin{eqnarray}
\partial_t f_1 &=& B^{(\delta)} f_1 + Q(f_1,f_1+f_2) \quad\mbox{and}\quad f_1(0,x,v)=f_0(x,v),\label{f1}
\\\partial_t f_2 &=& G f_2 + Q(f_2,f_2) + A^{(\delta)}f_1 \quad\mbox{and}\quad f_2(0,x,v)=0\label{f2}
\end{eqnarray}
with either specular reflections or Maxwellian diffusion boundary conditions.

\bigskip
The equation in the smaller space $\eqref{f2}$ will be treated thanks to the previous study in the $L^\infty\pa{\langle v \rangle^\beta\mu^{-1/2}}$ whilst we expect an exponential decay for solutions in the larger space $\eqref{f1}$. Indeed, $B^{(\delta)}$ can be controlled by the multiplicative operator $\nu(v)$ because it has a small norm in the following sense.

\bigskip
\begin{lemma}\label{lem:controlB2}
Consider $q$ in $\br{1,\infty}$. Let $m=e^{\kappa\abs{v}^\alpha}$ with $\kappa >0$ and $\alpha$ in $(0,2)$ or $m=\langle v \rangle^k$ with $k>k_q^*$ where
\begin{equation}\label{kq*}
k^*_q = \pa{\frac{16\pi b_\infty}{l_b}-2}^{1/q}\pa{1+\gamma + \frac{16\pi b_\infty}{l_b}}^{1-1/q}.
\end{equation}
Then $B^{(\delta)}_2$ satisfies
$$\forall h \in L^q_vL^\infty_x\pa{m}, \quad \norm{B^{(\delta)}_2(h)}_{L^q_vL^\infty_x\pa{\nu^{-1} m}} \leq \Delta_{m,q}(\delta)\norm{h}_{L^q_vL^\infty_x\pa{m}},$$
where $\nu(v)$ is the collision frequency $\eqref{LnuK}$ and $\Delta_{m,q}(\delta)$ is a constructive constant such that
\begin{itemize}
\item if $m=e^{\kappa\abs{v}^\alpha}$ then
$$\lim\limits_{\delta\to 0}\Delta_{m,q}(\delta) = 0$$
\item if $m=\langle v \rangle^k$ then
$$\lim\limits_{\delta \to 0}\Delta_{m,q}(\delta) = \phi_q(k) =\frac{16\pi b_\infty}{l_b}\left(\frac{1}{k+2}\right)^{1/q}\left(\frac{1}{k-1-\gamma}\right)^{1-1/q}.$$
\end{itemize}
\end{lemma}
\bigskip

This was proved in \cite[Lemma 4.12]{GMM} in the case if hard sphere ($\gamma=b=1$) and extended to more general hard potential with cutoff kernels in \cite[Lemma 6.3]{BriDau} for $L^2_v$ and \cite[Lemma 2.4]{Bri9} for $L^\infty_v$.

\bigskip
\begin{remark}\label{rem:kq*}
We point out that for $k>k_q^*$ one can check that $\phi_q(k) <1$. This will be of great importance for $B^{(\delta)}_2$ to be controlled by the semigroup generated by collision frequency $S_{G_\nu}(t)$.
\end{remark}
\bigskip

We conclude this subsection with a control on the bilinear term in the $L^\infty_{x,v}$ setting. 
\bigskip
\begin{lemma}\label{lem:controlQ}
For all $h$ and $g$ such that $Q(h,g)$ is well-defined, $Q(h,g)$ belongs to $\cro{\mbox{Ker}(L)}^\bot$ in $L^2_v$:
\begin{equation*}
\pi_L\pa{Q(h,g)}=0.
\end{equation*}
Moreover, let $q$ be in $\br{1,\infty}$ and let  $m=e^{\kappa\abs{v}^\alpha}$ with $\kappa >0$ and $\alpha$ in $(0,2)$ or $m=\langle v \rangle^k$ with $k\geq 0$. Then there exists $C_Q >0$ such that for all $h$ and $g$,
\begin{equation*}
\norm{Q(h,g)}_{L^q_vL^\infty_x\pa{m \nu^{-1}}} \leq C_Q \norm{h}_{L^q_vL^\infty_x\pa{m}}\norm{g}_{L^q_vL^\infty_x\pa{m}}.
\end{equation*}
The constant $C_Q$ is explicit and depends only on $q$, $m$ and the kernel of the collision operator.
\end{lemma}

\bigskip
\begin{proof}[Proof of Lemma $\ref{lem:controlQ}$]
Since we use the symmetric definition of $Q$ $\eqref{Qfg}$ the orthogonality property can be found in \cite{Bri1} Appendix $A.2.1$.
\par The estimate follows directly from \cite{GMM} Lemma $5.16$  and the fact that $\nu(v) \sim m(v)$ (see $\eqref{nu0nu1}$).
\end{proof}
\bigskip


\subsection{Study of equation $\eqref{f1}$ in $L^\infty_{x,v}\pa{m}$}\label{subsec:f1}

In the section we study the differential equation $\eqref{f1}$. We prove well-posedness for this problem and above all exponential decay as long as the initial data is small. The case of specular reflections and the case of diffusion are rather different since their treatment relies on the representation of the semigroup $S_{G_\nu}$ we derived in Section $\ref{sec:semigroupcollisionfrequency}$.


\subsubsection{The case of specular reflections in $L^\infty_{x,v}\pa{m}$}\label{subsubsec:f1SR}

We prove the following well-posedness result in the case of specular reflections.

\begin{prop}\label{prop:f1SR}
Let $\Omega$ be a $C^1$ bounded domain. Let $m=e^{\kappa\abs{v}^\alpha}$ with $\kappa >0$ and $\alpha$ in $(0,2)$ or $m=\langle v \rangle^k$ with $k> 5+\gamma$. Let $f_0$ be in $L^\infty_{x,v}\pa{m}$ and $g(t,x,v)$ in $L^\infty_{x,v}\pa{m}$. Then there exists $\delta_m$, $\lambda_m(\delta)>0$ such that for any $\delta$ in $(0,\delta_m]$ there exist $C_1$, $\eta_1>0$ such that if
$$ \norm{f_0}_{L^\infty_{x,v}\pa{m}}\leq \eta_1 \quad\mbox{and}\quad \norm{g}_{L^\infty_t L^\infty_{x,v}\pa{m}}\leq \eta_1,$$
then there exists a solution $f_1$ to 
\begin{equation}\label{f1propSR}
\partial_tf_1 = G_\nu f_1 + B^{(\delta)}_2f_1 + Q(f_1,f_1+g),
\end{equation}
with initial data $f_0$ and satisfying the specular reflections boundary conditions $\eqref{SR}$. Moreover, this solution satisfies
$$\forall t\geq 0,\quad \norm{f_1(t)}_{L^\infty_{x,v}\pa{m}}\leq C_1 e^{-\lambda_m(\delta) t}\norm{f_0}_{L^\infty_{x,v}\pa{m}},$$
and also
$$\lim\limits_{\delta \to 0} \lambda_{e^{\kappa\abs{v}^\alpha}}(\delta)=\nu_0 \quad\mbox{and}\quad \lim\limits_{k\to\infty}\lim\limits_{\delta\to 0} \lambda_{\langle v \rangle^k}(\delta) = \nu_0.$$
The constants $C_1$ and $\eta_1$ are constructive and only depend on $m$, $\delta$ and the kernel of the collision operator.
\end{prop}
\bigskip

\begin{proof}[Proof of Proposition $\ref{prop:f1SR}$]
If $f_1$ is solution to $\eqref{f1propSR}$ then, thanks to Proposition $\ref{prop:semigroupGnuSR}$, $G_\nu$ combined with boundary conditions generates a semigroup $S_{G_\nu}(t)$ in $L^\infty_{x,v}\pa{m}$. Therefore $f_1$ has the following Duhamel representation almost everywhere in $\R^+\times\Omega\times\R^3$
\begin{equation}\label{duhamelf1}
f_1(t,x,v) = S_{G_\nu}(t)f_0 + \int_0^t S_{G_\nu}(t-s)\cro{B^{(\delta)}_2 f_1+Q(f_1,f_1+g)}\:ds.
\end{equation}

\bigskip
To prove existence and exponential decay we use the following iteration scheme starting from $h_0 = 0$.
$$\left\{\begin{array}{l} \disp{h_{l+1} = S_{G_\nu}(t)f_0 + \int_0^t S_{G_\nu}(t-s)\cro{B^{(\delta)}_2 h_{l+1} + Q(h_l,h_l+g)}\:ds } \vspace{2mm}\\ \disp{h_{l+1}(0,x,v) = f_0(x,v).}\end{array}\right.$$
A contraction argument with the Duhamel representation would imply that $\pa{h_l}$ is well-defined in $L^\infty_{x,v}\pa{m}$ and satisfies specular reflections boundary condition (because $S_{G_\nu}$ does). The computations to prove this contraction are similar to the ones we make to prove that $\pa{h_l}$ is a Cauchy sequence and we therefore only write down the latter.
\begin{equation*}
\begin{split}
\abs{h_{l+1}-h_l}(t,x,v) \leq& \int_0^t \abs{S_{G_\nu}(t-s)B^{(\delta)}_2\pa{h_{l+1}-h_l}}\:ds
\\&+ \int_0^t \abs{S_{G_\nu}(t-s)Q\pa{h_l-h_{l-1},h_l+h_{l-1}+g}}\:ds.
\end{split}
\end{equation*}
 
\par For almost all $(t,x,v)$ using the representation of $S_{G_\nu}(t)$ $\eqref{representationGnuSR}$ with $\eps=0$ there exists $(X(t,x,v),V(t,x,v))$ in $\Omega\times\R^3$ such that the backward characteristics starting at $(x,v)$ reaches the initial plane $\br{t=0}$ at $(X(t,x,v),V(t,x,v))$ and
$$S_{G_{\nu}}(t)h = e^{-\nu(v)t}h(X(t,x,v),V(t,x,v)).$$
This implies that for almost all $(t,x,v)$
\begin{equation}\label{f1SRdifferenceCauchy}
\begin{split}
m(v)\abs{h_{l+1}-h_l}(t,x,v) &\leq \int_0^t e^{-\nu(v)(t-s)}m\abs{B^{(\delta)}_2\pa{h_{l+1}-h_l}(s,X(t-s),V(t-s))}ds
\\&\:+\int_0^t e^{-\nu(v)(t-s)}m\abs{Q\pa{h_l-h_{l-1},h_l+h_{l-1}+g}(s,X,V)}ds
\\&= I_1+I_2.
\end{split}
\end{equation}

\bigskip
$I_1$ is dealt with using Lemma $\ref{lem:controlB2}$,
\begin{eqnarray*}
I_1 &\leq& \int_0^t \nu(v)e^{-\nu(v)(t-s)}\norm{{B^{(\delta)}_2\pa{h_{l+1}-h_l}(s)}}_{L^\infty_{x,v}\pa{m\nu^{-1}}}\:ds
\\&\leq& \Delta_{m,\infty}(\delta)\int_0^t \nu(v)e^{-\nu(v)(t-s)}\norm{(h_{l+1}-h_l)(s)}_{L^\infty_{x,v}\pa{m}}\:ds.
\end{eqnarray*}
For $\delta$ small enough we have $\Delta_{m,\infty}(\delta)<1$, as emphasized in Remark $\ref{rem:kq*}$. Therefore,
\begin{equation}\label{epsSR}
\exists\: \eps \in (0,1), \quad \eps < 1-\Delta_{m,\infty}(\delta).
\end{equation}
Since $0< \eps <1$ it follows
$$\forall 0\leq s \leq t,\quad -\nu(v)(t-s) \leq -\eps\nu_0 t -\nu(v)(1-\eps)(t-s) +\eps\nu_0 s.$$
We can further bound $I_1$,
\begin{equation}\label{I1SR}
\begin{split}
I_1 &\leq e^{-\eps\nu_0 t} \Delta_{m,\infty}(\delta)\pa{\int_0^t \nu(v)e^{-\nu(v)(1-\eps)s}\:ds}\sup\limits_{0\leq s \leq t}\pa{e^{\eps\nu_0s}\norm{h_{l+1}-h_l}_{L^\infty_{x,v}\pa{m}}}
\\&\leq e^{-\eps \nu_0 t} \frac{\Delta_{m,\infty}(\delta)}{1-\eps}\sup\limits_{0\leq s \leq t}\pa{e^{\eps \nu_0 s}\norm{h_{l+1}-h_l}_{L^\infty_{x,v}\pa{m}}}.
\end{split}
\end{equation}

\bigskip
For the second term $I_2$ we multiply by $\nu(v)\nu(v)^{-1}$ to compensate for the loss of weight $\nu(v)$ in the control of $Q$. Then, with previous computations and using Lemma $\ref{lem:controlQ}$ this yields
\begin{equation}\label{I2SR}
\begin{split}
I_2 \leq& \frac{C_Q}{1-\eps} e^{-\eps \nu_0 t}\sup\limits_{0\leq s\leq t}\pa{e^{\eps\nu_0 s}\norm{h_l-h_{l-1}}_{L^\infty_{x,v}\pa{m}}}
\\&\quad\quad\quad\times\cro{\norm{h_l}_{L^\infty_{t,x,v}\pa{m}}+\norm{h_{l-1}}_{L^\infty_{t,x,v}\pa{m}}+\norm{g}_{L^\infty_{t,x,v}\pa{m}}}.
\end{split}
\end{equation}

\bigskip
We plug $\eqref{I1SR}$ and $\eqref{I2SR}$ into $\eqref{f1SRdifferenceCauchy}$ and multiply it by $e^{\eps\nu_0 t}$ before taking the supremum in $x$, $v$ and $t$. It yields
\begin{equation*}
\begin{split}
\pa{1-\frac{\Delta_{m,\infty}(\delta)}{1-\eps}}&\sup\limits_{0\leq s \leq t}\pa{e^{\eps\nu_0 s}\norm{h_{l+1}-h_{l}}_{L^\infty_{x,v}\pa{m}}} 
\\&\quad\quad\quad\leq   \frac{C_Q}{1-\eps}\cro{\norm{h_l}_{L^\infty_{t,x,v}\pa{m}}+\norm{h_{l-1}}_{L^\infty_{t,x,v}\pa{m}}+\norm{g}_{L^\infty_{t,x,v}\pa{m}}} 
\\&\quad\quad\quad\quad\quad\quad\times\sup\limits_{0\leq s\leq t}\pa{e^{\eps\nu_0 s}\norm{h_l-h_{l-1}}_{L^\infty_{x,v}\pa{m}}}.
\end{split}
\end{equation*}
Our choice of $\eps$ $\eqref{epsSR}$ implies $\cro{1-\Delta_{m,\infty}(\delta)(1-\eps)^{-1}} > 0$. Denoting by $C_m$ any positive constant independent of $l$ it follows
\begin{equation}\label{f1cauchyfinal}
\begin{split}
&\sup\limits_{0\leq s \leq t}\cro{e^{\eps\nu_0 s}\norm{h_l-h_{l+1}}_{L^\infty_{x,v}\pa{m}}}
\\&\leq C_m\cro{\norm{h_l}_{L^\infty_{t,x,v}\pa{m}}+\norm{h_{l-1}}_{L^\infty_{t,x,v}\pa{m}}+\norm{g}_{L^\infty_{t,x,v}\pa{m}}} \sup\limits_{0\leq s \leq t}\cro{e^{\eps\nu_0 s}\norm{h_l-h_{l-1}}_{L^\infty_{x,v}\pa{m}}}.
\end{split}
\end{equation}

\bigskip
We now prove that $\norm{h_l}_{L^\infty_{x,v}\pa{m}}$ is uniformly bounded. Starting with the definition of $h_{l+1}$ and making the same computations without subtracting $h_l$ we obtain
\begin{equation*}
\begin{split}
&\pa{1-\frac{\Delta_{m,\infty}(\delta)}{1-\eps}}\sup\limits_{0\leq s \leq t}\pa{e^{\eps\nu_0 s}\norm{h_{l+1}}_{L^\infty_{x,v}\pa{m}}} 
\\&\quad\quad\quad\leq e^{-\nu_0(1-\eps)t}\norm{f_0}_{L^\infty_{x,v}}+ C_Q\frac{\Delta_{m,\infty}(\delta)}{1-\eps}\sup\limits_{0\leq s \leq t}\pa{e^{\eps\nu_0 s}\norm{h_l}_{L^\infty_{[0,t],x,v}\pa{m}}}
\\&\quad\quad\quad\quad\quad\quad\quad\quad\quad\quad\quad\quad\quad\quad\times\pa{\norm{h_l}_{L^\infty_{[0,t],x,v}\pa{m}} + \norm{g}_{L^\infty_{[0,t],x,v}\pa{m}}},
\end{split}
\end{equation*}
where we used the exponential decay of $S_{G_\nu}(t)$ on $f_0$ (Proposition $\ref{prop:semigroupGnuSR}$). Denoting $C^{(1)}_m$ and $C^{(2)}_m$ any positive constant independent of $l$, we further bound
\begin{equation}\label{f1boundfinal}
\begin{split}
&\sup\limits_{0\leq s \leq t}\pa{e^{\eps\nu_0 s}\norm{h_{l+1}}_{L^\infty_{x,v}\pa{m}}}
\\&\leq C^{(1)}_m\norm{f_0}_{L^\infty_{x,v}} + C^{(2)}_m \sup\limits_{0\leq s \leq t}\pa{e^{\eps\nu_0 s}\norm{h_l}_{L^\infty_{[0,t],x,v}\pa{m}}}\cro{\norm{h_l}_{L^\infty_{[0,t],x,v}\pa{m}} + \norm{g}_{L^\infty_{[0,t],x,v}\pa{m}}}.
\end{split}
\end{equation}
Therefore, if $\norm{h_0}_{L^\infty_{x,v}}$ and $\norm{g}_{L^\infty_{t,x,v}}$ are smaller than $\eta_1>0$ such that
$$C^{(1)}_m \eta_1 + 2\pa{1+C^{(1)}_m}^2 C^{(2)}_m\eta_1^2 \leq \pa{1+C^{(1)}_m}\eta_1$$
then for all $l\in \N$
\begin{equation}\label{f1decaySR}
\sup\limits_{t\geq 0}\pa{e^{\eps\nu_0 s}\norm{h_{l+1}}_{L^\infty_{x,v}\pa{m}}} \leq \pa{1+C^{(1)}_m}\norm{f_0}_{L^\infty_{x,v}},
\end{equation}
which gives the desired exponential decay if $\pa{h_l}_{l \in \N}$ converges.

\bigskip
Combining $\eqref{f1cauchyfinal}$ and $\eqref{f1decaySR}$ we have
$$\sup\limits_{0\leq s \leq t}\pa{e^{\eps\nu_0 s}\norm{h_{l+1}-h_l}_{L^\infty_{t,x,v}(m)}} \leq 3C_m\pa{1+C^{(1)}_m}\eta_1 \sup\limits_{0\leq s \leq t}\pa{e^{\eps\nu_0 s}\norm{h_l-h_{l-1}}_{L^\infty_{x,v}(m)}}.$$
Therefore, for $\eta_1$ small enough the sequence $\pa{h_l}_{l \in \N}$ is a Cauchy sequence in $L^\infty_{t,x,v}(m)$ and therefore converges towards $f_1$. From $\eqref{f1decaySR}$, $f_1$ satisfies the desired exponential decay with $\lambda_m\pa{\delta} = \eps\nu_0$.
\par The asymptotic behaviour of $\lambda_m\pa{\delta}$ is straightforward since as $\Delta_{m,\infty}(\delta)$ is closer to $0$ we can choose $\eps$ closer to $1$ by $\eqref{epsSR}$.
\end{proof}
\bigskip


\subsubsection{The case of Maxwellian diffusion in $L^\infty_{x,v}\pa{e^{\kappa\abs{v}^\alpha}}$}\label{subsubsec:f1MD}

We prove the following well-posedness result for $\eqref{f1}$ in the case of Maxwellian diffusion.

\begin{prop}\label{prop:f1MD}
Let $\Omega$ be a $C^1$ bounded domain. Let $m=e^{\kappa\abs{v}^\alpha}$ with $\kappa >0$ and $\alpha$ in $(0,2)$.Let $f_0$ be in $L^\infty_{x,v}\pa{m}$ and $g(t,x,v)$ in $L^\infty_{x,v}\pa{m}$. Then there exists $\delta_m$, $\lambda_m(\delta)>0$ such that for any $\delta$ in $(0,\delta_m]$ there exist $C_1$, $\eta_1>0$ such that if
$$ \norm{f_0}_{L^\infty_{x,v}\pa{m}}\leq \eta_1 \quad\mbox{and}\quad \norm{g}_{L^\infty_tL^\infty_{x,v}\pa{m}}\leq \eta_1,$$
then there exists a solution $f_1$ to 
\begin{equation}\label{f1propMD}
\partial_tf_1 = G_\nu f_1 + B^{(\delta)}_2f_1 + Q(f_1,f_1+g),
\end{equation}
with initial data $f_0$ and satisfying the Maxwellian diffusion boundary conditions $\eqref{MD}$. Moreover, this solution satisfies
$$\forall t\geq 0,\quad \norm{f_1(t)}_{L^\infty_{x,v}\pa{m}}\leq C_1 e^{-\lambda_m(\delta) t}\norm{f_0}_{L^\infty_{x,v}\pa{m}},$$
and also
$$\lim\limits_{\delta \to 0} \lambda_m(\delta)=\nu_0.$$
The constants $C_1$ and $\eta_1$ are constructive and only depend on $m$, $\delta$ and the kernel of the collision operator.
\end{prop}
\bigskip

\begin{proof}[Proof of Proposition $\ref{prop:f1MD}$]
Thanks to Proposition  $\ref{prop:semigroupGnuMD}$, $G_\nu$ combined with Maxwellian diffusion boundary conditions generates a semigroup $S_{G_\nu}(t)$ in all the $L^\infty_{x,v}\pa{m}$. Therefore a solution $f_1$ to $\eqref{f1propMD}$ has the following Duhamel representation almost everywhere in $\R^+\times\Omega\times\R^3$
\begin{equation}\label{duhamelf1MD}
f_1(t,x,v) = S_{G_\nu}(t)f_0 + \int_0^t S_{G_\nu}(t-s)\cro{B^{(\delta)}_2 f_1+Q(f_1,f_1+g)}\:ds.
\end{equation}

\bigskip
We use the same iteration as for specular reflections, starting from $h_0=0$ and defining
$$\left\{\begin{array}{l} \disp{h_{l+1} = S_{G_\nu}(t)f_0 + \int_0^t S_{G_\nu}(t-s)\cro{B^{(\delta)}_2 h_{l+1} + Q(h_l,h_l+g)}\:ds } \vspace{2mm}\\ \disp{h_{l+1}(0,x,v) = f_0(x,v).}\end{array}\right.$$
Again, the well-posedness of $h_{l+1}$ follows a contraction argument with the Duhamel representation and the estimates we shall prove in order to show that $\pa{h_l}_{l \in \N}$ is a Cauchy sequence with uniform exponential decay. We therefore only prove the latter.
\par Using the implicit representation of $S_{G_\nu}(t)$ $\eqref{startinductionMD}$-$\eqref{inductionMD}$ (note that we do not have the change of weight) we have for $h$ in $L^\infty_{x,v}\pa{m}$:
\begin{itemize}
\item if $t_1\leq 0$ then 
\end{itemize}
\begin{equation}\label{startinductionMDf1}
S_{G_\nu}(t)h (x,v) = e^{-\nu(v)t}h(x-tv,v);
\end{equation}
\begin{itemize}
\item if $t_1 > 0$ then for all $p\geq 2$, 
\end{itemize}
\begin{equation}\label{inductionMDf1}
\begin{split}
S_{G_\nu}(t)h (x,v) &= c_\mu \mu(v) e^{-\nu(v)(t-t_1)}\sum\limits_{i=1}^{p}\int_{\prod\limits_{j=1}^{p}\br{v_j \cdot n(x_i)>0}}\mathbf{1}_{[t_{i+1},t_i)}(0)\:h(x_i - t_iv_i,v_i)d\Sigma_i(0)
\\&\quad +c_\mu \mu(v)e^{-\nu(v)(t-t_1)}\int_{\prod\limits_{j=1}^{p}\br{v_j \cdot n(x_i)>0}}\mathbf{1}_{t_{p+1}>0}\:S_{G_\nu}(t_p)h(x_p,v_p)d\Sigma_{p}(t_p),
\end{split}
\end{equation}
where
\begin{equation}\label{dSigmai}
d\Sigma_{i}(s) = \frac{1}{c_\mu \mu(v_i)}e^{-\nu(v_i)(t_i-s)}\pa{\prod\limits_{j=1}^{i-1}e^{-\nu(v_j)(t_j-t_{j+1})}} \:d\sigma_{x_1}(v_1)\dots d\sigma_{x_{p}}(v_{p}).
\end{equation}

\bigskip
We shall prove that $\pa{h_l}_{l\in\N}$ is a Cauchy sequence in $L^\infty_tL^\infty_{x,v}\pa{m}$. We bound $\abs{h_{l+1}-h_l}$ in $L^\infty_{x,v}\pa{m}$ by
\begin{equation*}
\begin{split}
&\norm{h_{l+1}-h_l}_{L^\infty_{x,v}\pa{m}} \leq \sup\limits_{(x,v)\in\Omega\times\R^3}\cro{\int_0^t  m(v) \abs{S_{G_\nu}(t-s)B^{(\delta)}_2\pa{h_{l+1}-h_l}}(x,v)\:ds}
\\&\quad\quad\quad\quad\quad+\sup\limits_{(x,v)\in\Omega\times\R^3}\cro{\int_0^t m(v)\abs{S_{G_\nu}(t-s)Q\pa{h_l-h_{l-1},h_l+h_{l-1}+g}}(x,v)\:ds}.
\end{split}
\end{equation*}

Since the behaviour of $S_{G_\nu}(t-s)$ differs whether $t_1\leq 0$ or $t_1> 0$, where $t_1=t_1(t-s,x,v)$, we can further decompose each of the terms on the right-hand side.
\begin{equation}\label{f1MDdifferenceCauchy}
\norm{h_{l+1}-h_l}_{L^\infty_{x,v}\pa{m}} \leq \max\br{I_1;I_3} + \max\br{I_2; I_4},
\end{equation}
where we defined
\begin{equation*}
\begin{split}
I_1 =& \sup\limits_{(x,v)\in\Omega\times\R^3}\cro{\int_0^t  m(v)\abs{S_{G_\nu}(t-s)B^{(\delta)}_2\pa{h_{l+1}-h_l}\mathbf{1}_{t_1\leq 0}}(x,v)\:ds}
\\I_2 =& \sup\limits_{(x,v)\in\Omega\times\R^3}\cro{\int_0^t m(v)\abs{S_{G_\nu}(t-s)Q\pa{h_l-h_{l-1},h_l+h_{l-1}+g}\mathbf{1}_{t_1\leq 0}}(x,v)\:ds}
\\I_3 =& \sup\limits_{(x,v)\in\Omega\times\R^3}\cro{\int_0^t  m(v)\abs{S_{G_\nu}(t-s)B^{(\delta)}_2\pa{h_{l+1}-h_l}\mathbf{1}_{t_1> 0}}(x,v)\:ds}
\\I_4 =& \sup\limits_{(x,v)\in\Omega\times\R^3}\cro{\int_0^t m(v)\abs{S_{G_\nu}(t-s)Q\pa{h_l-h_{l-1},h_l+h_{l-1}+g}\mathbf{1}_{t_1 > 0}}(x,v)\:ds}.
\end{split}
\end{equation*}
We fix $\eps$ in $(0,1)$.

\bigskip
\textbf{Study of $\mathbf{I_1}$ and $\mathbf{I_2}$.}
When $t_1(t-s,x,v)\leq 0$ the semigroup $S_{G_\nu}(t-s)$ is a mere multiplication by $e^{-\nu(v)(t-s)}$ and so
$$I_1 \leq \sup\limits_{v\in\R^3}\int_0^t \nu(v)e^{-\nu(v)(t-s)} \sup\limits_{x\in\Omega}\cro{m(v)\nu(v)^{-1}B^{(\delta)}_2\pa{h_{l+1}-h_l}}\:ds$$
and equivalently for $I_2$. Similar computations as $\eqref{I1SR}$-$\eqref{I2SR}$ yields

\begin{eqnarray}
I_1 &\leq& e^{-\eps\nu_0 t} \frac{\Delta_{m,\infty}(\delta)}{1-\eps}\:\sup\limits_{0\leq s \leq t}\cro{e^{\eps\nu_0 s}\norm{h_{l+1}-h_l}_{L^\infty_{x,v}\pa{m}}} \label{I1MD}
\\I_2 &\leq& \frac{C_Q}{1-\eps} e^{-\eps\nu_0 t}\:\sup\limits_{0\leq s\leq t}\cro{e^{\eps\nu_0 s}\norm{h_l-h_{l-1}}_{L^\infty_{x,v}\pa{m}}}\label{I2MD}
\\&&\quad\quad\times\cro{\norm{h_l}_{L^\infty_{t,x,v}\pa{m}}+\norm{h_{l-1}}_{L^\infty_{t,x,v}\pa{m}}+\norm{g}_{L^\infty_{t,x,v}\pa{m}}} . \nonumber
\end{eqnarray}

\bigskip
\textbf{Study of $\mathbf{I_3}$ and $\mathbf{I_4}$.} We study $I_3$ and $I_4$ are dealt with the same way and we therefore only write down the details for $I_3$.
\par We decompose 
\begin{equation}\label{decompositionI3}
I_3 \leq I_3^{(1)}+I_3^{(2)}
\end{equation}
into two terms defined by $\eqref{inductionMDf1}$:
\begin{eqnarray*}
I_3^{(1)} &=& \sup\limits_{(x,v) \in \Omega\times\R^3}\Big\{\int_0^t c_\mu \mu(v)m(v) e^{-\nu(v)(t-s-t_1)}
\\&& \quad\times\sum\limits_{i=1}^{p}\int_{\prod\limits_{j=1}^{p}\br{v_j \cdot n(x_i)>0}}\mathbf{1}_{[t_{i+1},t_i)}(0)\:\abs{B^{(\delta)}_2\pa{h_{l+1}-h_l}}(s,x_i-t_iv_i,v_i)d\Sigma_i(0)\Big\}
\\I_3^{(2)} &=& \sup\limits_{(x,v) \in \Omega\times\R^3}\Big\{\int_0^t  c_\mu \mu(v)m(v)e^{-\nu(v)(t-t_1)}
\\&&\quad\times\int_{\prod\limits_{j=1}^{p}\br{v_j \cdot n(x_i)>0}}\mathbf{1}_{t_{p+1}>0}\:\abs{S_{G_\nu}(t_p)\pa{B^{(\delta)}_2\pa{h_{l+1}-h_l}}}(x_p,v_p)d\Sigma_{p}(t_p)\Big\}.
\end{eqnarray*}

\bigskip
We multiply and divide by $m(v_i)\nu^{-1}(v_i)$ and we take the supremum over $\Omega\times\R^3$ inside the $i^{th}$ integral. We know that $c_\mu\mu(v)m(v)$ is bounded in $v$ and therefore, denoting by $C_m$ any positive constant, we obtain
\begin{equation}
\begin{split}
I_3^{(1)} \leq& C_m \int_0^t e^{-\nu_0(t-s)}\norm{B_2^{\delta}(h_{l+1}-h_l)(s)}_{L^\infty_{x,v}\pa{m\nu^{-1}}} \label{I31MDstart}
\\&\times\sup\limits_{(x,v) \in \Omega\times\R^3}\sum\limits_{i=1}^{p}\int_{\prod\limits_{\overset{j=1}{j\neq i}}^{p}\br{v_j \cdot n(x_i)>0}}\mathbf{1}_{[t_{i+1},t_i)}(0)\pa{\int_{\R^3}\frac{\abs{v_i} \nu(v_i)}{m(v_i)}\:dv_i} \:d\sigma_{x_i}ds.
\end{split}
\end{equation}
We recall that $d\sigma_{x_i}$ is a probability measure on $\br{v_j\cdot n(x_i)>0}$. The integral in the variable $v_i$ is finite and only depends on $m$ and $\nu$. Using Lemma $\ref{lem:controlB2}$ we conclude
\begin{eqnarray}
I_3^{(1)} &\leq & C_m \Delta_{m,\infty}(\delta) \int_0^t e^{-\nu_0(t-s)}\norm{h_{l+1}(s)-h_l(s)}_{L^\infty_{x,v}\pa{m}}\:ds \nonumber
\\ &\leq & \frac{C_m \Delta_{m,\infty}(\delta)}{1-\eps}e^{-\eps\nu_0 t}\sup\limits_{0\leq s \leq t}\cro{e^{\eps \nu_0 s}\norm{h_{l+1}(s)-h_l(s)}_{L^\infty_{x,v}\pa{m}}}. \label{I31MD}
\end{eqnarray}

\bigskip
To estimate $I_3^{(2)}$ we first see that, as noticed in $\eqref{tp+1tot1}$,
\begin{equation*}
\begin{split}
&\mathbf{1}_{t_{p+1}>0}\:\abs{S_{G_\nu}(t_p)\pa{B^{(\delta)}_2\pa{h_{l+1}-h_l}}}(x_p,v_p) 
\\&\quad\quad\quad\leq \mathbf{1}_{t_{p}>0}\:\abs{S_{G_\nu}(t_p)\pa{B^{(\delta)}_2\pa{h_{l+1}-h_l}}\mathbf{1}_{t_1(t_p,x_p,v_p)}}(x_p,v_p).
\end{split}
\end{equation*}
Thanks to Corollary $\ref{cor:SGnuMDcontrolnu-1}$ with $\nu_0' = (1-\eps')\nu_0$, where $0<\eps<\eps'<1$, and then Lemma $\ref{lem:controlB2}$ we can estimate the above further by
\begin{eqnarray*}
\mathbf{1}_{t_{p+1}>0}\:\abs{S_{G_\nu}(t_p)\pa{B^{(\delta)}_2\pa{h_{l+1}-h_l}}} &\leq& C_m e^{-\nu'_0 t_p}\norm{B^{(\delta)}_2\pa{h_{l+1}-h_l}}_{L^\infty_{x,v}\pa{m\nu^{-1}}}
\\&\leq& C_m \Delta_{m,\infty}(\delta)e^{-(1-\eps') \nu_0 t_p}\norm{h_{l+1}-h_l}_{L^\infty_{x,v}\pa{m}}.
\end{eqnarray*}

Plugging the above into the definition of $I^{(2)}_3$ yields
\begin{equation}\label{I32MD}
\begin{split}
I^{(2)}_3 \leq & \:C_m\Delta_{m,\infty}(\delta) e^{-\nu_0(t-t_1)}
\\&\times\int_0^t \sup\limits_{(x,v)\in\Omega\times\R^3}\int_{\prod\limits_{j=1}^{p}\br{v_j \cdot n(x_i)>0}}e^{-\nu_0(t_1-t_p)}e^{-(1-\eps')\nu_0 t_p}\norm{h_{l+1}-h_l}_{L^\infty_{x,v}\pa{m}}(t_p)d\sigma_{x_i}.
\\\leq & \:C_m \Delta_{m,\infty}(\delta)t e^{-\eps'\nu_0 t}\sup\limits_{0\leq s \leq t}\cro{e^{\eps\nu_0 s}\norm{h_{l+1}-h_l}_{L^\infty_{x,v}\pa{m}}}
\\\leq & \:C_{m} \Delta_{m,\infty}(\delta)e^{-\eps\nu_0 t}\sup\limits_{0\leq s \leq t}\cro{e^{\eps\nu_0 s}\norm{h_{l+1}-h_l}_{L^\infty_{x,v}\pa{m}}}.
\end{split}
\end{equation}

\bigskip
We conclude the estimate about $I_3$ by gathering $\eqref{I31MD}$ and $\eqref{I32MD}$ inside the decomposition of $I_3$ $\eqref{decompositionI3}$:
\begin{equation}\label{I3MD}
I_3 \leq C_{m}\Delta_{m,\infty}(\delta)e^{-\eps\nu_0t}\sup\limits_{0\leq s \leq t}\cro{e^{\eps\nu_0 s}\norm{h_{l+1}-h_l}_{L^\infty_{x,v}\pa{m}}}
\end{equation}

For the term $I_4$ we can do exactly the same computations with $\Delta_{m,\infty}(\delta)$ replaced by $C_Q$ from Lemma $\ref{lem:controlQ}$, which can be included into the generic constant $C_{m}$. Hence
\begin{equation}\label{I4MD}
\begin{split}
I_4 \leq &\:C_{m}e^{-\eps\nu_0t}\sup\limits_{0\leq s \leq t}\cro{e^{\eps\nu_0 s}\norm{h_{1}-h_{l-1}}_{L^\infty_{x,v}\pa{m}}}
\\&\:\quad\quad\quad\times\cro{\norm{h_l}_{L^\infty_{t,x,v}\pa{m}}+\norm{h_{l-1}}_{L^\infty_{t,x,v}\pa{m}}+\norm{g}_{L^\infty_{t,x,v}\pa{m}}} 
\end{split}
\end{equation}

\bigskip
\textbf{Conclusion.} From the decomposition $\eqref{f1MDdifferenceCauchy}$ of $h_{l+1}-h_l$ and estimates $\eqref{I1MD}$-$\eqref{I2MD}$-$\eqref{I3MD}$-$\eqref{I4MD}$ we obtain
\begin{equation*}
\begin{split}
&\pa{1-C_m\Delta_{m,\infty}(\delta)}\sup\limits_{0\leq s \leq t}\norm{h_{l+1}-h_l}_{L^\infty_{x,v}\pa{m}}
\\&\quad\quad\leq C_m \cro{\norm{h_l}_{L^\infty_{t,x,v}\pa{m}}+\norm{h_{l-1}}_{L^\infty_{t,x,v}\pa{m}}+\norm{g}_{L^\infty_{t,x,v}\pa{m}}}
\\&\quad\quad\quad\times\pa{\sup\limits_{0\leq s \leq t}\cro{e^{\varepsilon\nu_0 s}\norm{h_{l+1}-h_l}_{L^\infty_{x,v}\pa{m}}}+\sup\limits_{0\leq s \leq t}\cro{e^{\varepsilon\nu_0 s}\norm{h_{l}-h_{l-1}}_{L^\infty_{x,v}\pa{m}}}}
\end{split}
\end{equation*}
In the case of a stretch exponential $m$, Lemma $\ref{lem:controlB2}$ states that $\Delta_{m,\infty}(\delta)$ tends to $0$ as $\delta$ goes to $0$. We can therefore choose $\delta$ small enough such that 
\begin{equation}\label{choiceDeltam}
1-C_m\Delta_{m,\infty}(\delta) \geq \frac{1}{2}.
\end{equation}
With such a choice the following holds
\begin{equation}\label{f1cauchyfinalMD}
\begin{split}
&\sup\limits_{0\leq s \leq t}\norm{h_{l+1}-h_l}_{L^\infty_{x,v}\pa{m}}
\\&\quad\quad\leq C_m \cro{\norm{h_l}_{L^\infty_{t,x,v}\pa{m}}+\norm{h_{l-1}}_{L^\infty_{t,x,v}\pa{m}}+\norm{g}_{L^\infty_{t,x,v}\pa{m}}}
\\&\quad\quad\quad\times\pa{\sup\limits_{0\leq s \leq t}\cro{e^{\varepsilon\nu_0 s}\norm{h_{l+1}-h_l}_{L^\infty_{x,v}\pa{m}}}+\sup\limits_{0\leq s \leq t}\cro{e^{\varepsilon\nu_0 s}\norm{h_{l}-h_{l-1}}_{L^\infty_{x,v}\pa{m}}}}
\end{split}
\end{equation}

\bigskip
Similar computations with the use of the exponential decay of $S_{G_\nu}(t)$ on $f_0$ gives us the following bound on $h_{l+1}$
\begin{equation}\label{f1boundfinalMD}
\begin{split}
&\sup\limits_{0\leq s \leq t}\cro{e^{\eps\nu_0 s}\norm{h_{l+1}}_{L^\infty_{x,v}\pa{m}}}
\\&\leq C^{(1)}_m\norm{f_0}_{L^\infty_{x,v}} + C^{(2)}_m \cro{\norm{h_l}_{L^\infty_{[0,t],x,v}\pa{m}} + \norm{g}_{L^\infty_{[0,t],x,v}\pa{m}}}\sup\limits_{0\leq s \leq t}\cro{e^{\eps\nu_0 s}\norm{h_l}_{L^\infty_{[0,t],x,v}\pa{m}}}.
\end{split}
\end{equation}

\bigskip
The latter results $\eqref{f1cauchyfinalMD}$ and $\eqref{f1boundfinalMD}$ are identical to respectively $\eqref{f1cauchyfinal}$ and $\eqref{f1cauchyfinal}$ in the case of specular reflections boundary conditions. Therefore the same arguments hold and if $\norm{f_0}_{L^\infty_{x,v}\pa{m}}$ and $\norm{g}_{L^\infty_{x,v}\pa{m}}$ are smaller than $\eta_1>0$ small enough we obtain that $\pa{h_l}_{l\in\N}$ is a Cauchy sequence and thus converges towards the desired solution $f_1$, which satisfies the required exponential decay. This concludes the proof of Proposition $\ref{prop:f1MD}$.

\end{proof}
\bigskip


\subsubsection{The case of Maxwellian diffusion in $L^\infty_{x,v}\pa{\langle v \rangle^k}$}\label{subsubsec:f1MDpoly}

Looking at the proof of Proposition $\ref{prop:f1MD}$ we remark that the key property used in the case of a stretch exponential weight $m$ is that $\Delta_{m,\infty}(\delta)$ tends to $0$ as $\delta$ goes to $0$. This strong property allowed us to control the supremum of $c_\mu \mu(v) m(v)$ in $I^{(1)}_3$ thanks to $\Delta_{m,\infty}(\delta)$ and still obtain a quantity that is less than $1$, see $\eqref{I31MD}$ and $\eqref{choiceDeltam}$.
\par Unfortunately, in the case of a polynomial weight $m_k(v) =\langle v \rangle^k$, Lemma $\ref{lem:controlB2}$ states that $\Delta_{m,\infty}(\delta)$ converges to a quantity less than $1$ but not as small as one wants unless one allows $k$ to be as large as one wants. However, $\Delta_{m_k,\infty}(\delta)$ goes to $0$ as $k$ tends to infinity like $k^{-1}$ which is not enough to control the supremum of $c_\mu \mu(v) m_k(v)$ which grows like $(2k)^k$.
\par The key idea to deal with the polynomial weight $m_k$ is that fact that $B^{(\delta)}_2$ can also be estimated in $L^1_vL^\infty_x\pa{\langle v \rangle^{2+\gamma+0}}$ (see Lemma $\ref{lem:controlB2}$). The latter norm is weaker than $L^\infty_{x,v}\pa{m_k}$ and appear in the estimate of $I^{(1)}_1$. Again, Lemma $\ref{lem:controlB2}$ still does not give the appropriate decay for $\Delta_{m_k,1}(\delta)$. The following lemma shows a new estimate on $B^{(\delta)}_2$ involving a mixing of the $L^1_vL^\infty_x$ and the $L^\infty_{x,v}$ frameworks.

\bigskip
\begin{lemma}\label{lem:mixingcontrolB2}
Let $k > 5+\gamma$ and $m_k(v)=\langle v \rangle^k$. Then for any $\delta >0$ there exists $\tilde{\Delta}_k(\delta)$ such that for all $h$ in $L^\infty_{x,v}\pa{m_k}$,
$$\norm{B^{(\delta)}_2 h}_{L^1_vL^\infty_x\pa{\langle v \rangle^2}} \leq \tilde{\Delta}_k(\delta) \norm{h}_{L^\infty_{x,v}\pa{m_k}}.$$
Moreover, the following holds for any $k>5+\gamma$
$$\lim\limits_{\delta \to 0} \tilde{\Delta}_k(\delta) = 0.$$
\end{lemma}
\bigskip

\begin{proof}[Proof of Lemma $\ref{lem:mixingcontrolB2}$]
We recall the definition of $B^{(\delta)}_2 h$,
$$B^{(\delta)}_{2} h (v) = \int_{\R^3\times\mathbb{S}^2}\left(1-\Theta_\delta\right)\left[\mu'_*h' + \mu'h'_* - \mu h_*\right]b\left(\mbox{cos}\:\theta\right)\abs{v-v_*}^\gamma\:d\sigma dv_*,$$
where $\Theta_\delta = \Theta_\delta(v,v_*,\sigma)$ is a $C^\infty$ function such that $0 \leq 1-\Theta_\delta \leq 1$ and such that $1-\Theta_\delta=0$ on the set
$$\Xi_\delta = \left\{\abs{v}\leq \delta^{-1}    \quad\mbox{and}\quad 2\delta\leq\abs{v-v_*}\leq \delta^{-1}    \quad\mbox{and}\quad \abs{\mbox{cos}\:\theta} \leq 1-2\delta \right\}.$$
We denote by $\Xi^c_\delta$ the complementary set of $\Xi_\delta$ in $\R^3\times\R^3\times\mathbb{S}^2$.

\bigskip
Only $h$ has a dependency in $x$ hence
$$\norm{B^{(\delta)}_2h}_{L^1_vL^\infty_x\pa{\langle v \rangle^2}} \leq \int_{\Xi_\delta^c} \cro{\mu'_*H' + \mu'H'_* + \mu H_*}\abs{v-v_*}^\gamma\abs{b\pa{\cos \theta}}\langle v \rangle^2\:dvdv_*d\sigma,$$
where we used the notation
$$H(v) = \sup\limits_{x \in \Omega} \abs{h(x,v)}.$$
We notice that $\Xi_\delta^c \subset \tilde{\Xi}_\delta$ where we defined
\begin{eqnarray}
\tilde{\Xi}_\delta &=& \br{\sqrt{\abs{v}^2+\abs{v_*}^2}\geq \frac{1}{\delta}} \cup \br{\abs{v-v_*}\leq 2\delta} \nonumber
\\&\:& \quad\cup \br{\abs{v-v_*} \geq \frac{1}{\delta}} \cup \br{1-\delta \leq \abs{\cos \theta}\leq 1}\nonumber
\\ &=& \tilde{\Xi}^{(1)}_\delta \cup \tilde{\Xi}^{(2)}_\delta \cup \tilde{\Xi}^{(3)}_\delta \cup \tilde{\Xi}^{(4)}_\delta \label{decompositiontildeXi}
\end{eqnarray}
and hence
$$\norm{B^{(\delta)}_2 h}_{L^1_vL^\infty_x\pa{\langle v\rangle^2}} \leq \int_{\tilde{\Xi}_\delta}\cro{\mu'_*H' + \mu'H'_* + \mu H_*}\abs{v-v_*}^\gamma\abs{b\pa{\cos \theta}}\langle v \rangle^2\:dvdv_*d\sigma.$$

\bigskip
The set $\tilde{\Xi}_\delta$ is invariant under the standard changes of variables $(v,v_*,\sigma) \mapsto (v_*,v,-\sigma)$ and $(v,v_*,\sigma) \mapsto (v',v'_*,k)$ with $k=\pa{v-v_*}/\abs{v-v_*}$ that have Jacobian $1$ (see \cite{CIP} or \cite{Vi2} for instance). Applying these change of variables gives
$$\norm{B^{(\delta)}_2h}_{L^1_vL^\infty_x\pa{\langle v \rangle^2}} \leq \int_{\tilde{\Xi}_\delta} \mu_* H\cro{\langle v_*' \rangle^2+\langle v' \rangle^2+\langle v \rangle^2}\abs{v-v_*}^\gamma\abs{b\pa{\cos \theta}}\:dvdv_*d\sigma.$$
Thanks to the elastic collisions one has
$$\langle v_*' \rangle^2+\langle v' \rangle^2 = \langle v_* \rangle^2+\langle v \rangle^2.$$
Therefore we have 
$$\norm{B^{(\delta)}_2h}_{L^1_vL^\infty_x\pa{\langle v \rangle^2}} \leq \tilde{\Delta}_k(\delta)\norm{h}_{L^\infty_{x,v}\pa{m_k}},$$
with
\begin{eqnarray*}
\tilde{\Delta}_k(\delta) &=& 3\int_{\tilde{\Xi}_\delta} \mu_* \langle v_* \rangle^2 \frac{\langle v \rangle^2}{m_k(v)} \abs{v-v_*}^\gamma\abs{b\pa{\cos \theta}} \:dvdv_*d\sigma 
\\&=&  3\int_{\tilde{\Xi}_\delta} \mu_* \langle v_* \rangle^2 \frac{1}{\langle v \rangle^{k-2}} \abs{v-v_*}^\gamma\abs{b\pa{\cos \theta}} \:dvdv_*d\sigma.
\end{eqnarray*}
It remains to show that when $k > 5+\gamma$ is fixed then $\tilde{\Delta}_k(\delta)$ goes to $0$ when $\delta$ goes to $0$.

\bigskip
We decompose the integral over $\tilde{\Xi}_\delta$ into integrals over $\tilde{\Xi}^{(1)}_\delta$, $\tilde{\Xi}^{(2)}_\delta$, $\tilde{\Xi}^{(3)}_\delta$ and $\tilde{\Xi}^{(4)}_\delta$ where these domains are given by $\eqref{decompositiontildeXi}$:
$$\tilde{\Delta}_k(\delta)= \tilde{\Delta}^{(1)}_k(\delta)+\tilde{\Delta}^{(2)}_k(\delta)+\tilde{\Delta}^{(3)}_k(\delta)+\tilde{\Delta}^{(4)}_k(\delta)$$

\par For $\tilde{\Delta}^{(1)}_k(\delta)$ and $\tilde{\Delta}^{(3)}_k(\delta)$ we bound crudely $\abs{v-v_*}$ by $\langle v \rangle \langle v_*\rangle$. Moreover, the inequality $\sqrt{\abs{v}^2+\abs{v_*}^2}\geq \delta^{-1}$ implies that $\abs{v}\geq 1/(2\delta)$ or $\abs{v_*}\geq 1/(2\delta)$. And the same holds for $\abs{v-v_*}\geq \delta^{-1}$. Therefore, for $i=1$ and $i=3$ we have
\begin{equation*}
\begin{split}
\tilde{\Delta}^{(i)}_k(\delta) \leq& l_b\pa{\int_{\abs{v_*}\geq \frac{1}{2\delta}}\mu_*\langle v_* \rangle^{2+\gamma}\:dv_*} \pa{\int_{\R^3}\frac{dv}{\langle v \rangle^{k-2-\gamma}}}
\\&\quad + l_b\pa{\int_{\R^3}\mu_*\langle v_* \rangle^{2+\gamma}\:dv_*} \pa{\int_{\abs{v}\geq \frac{1}{2\delta}}\frac{dv}{\langle v \rangle^{k-2-\gamma}}},
\end{split}
\end{equation*}
where $l_b$ is the integral of $b\pa{\cos \theta}$ on $\mathbb{S}^2$. Since $k>5+\gamma$, $\langle v \rangle^{k-2-\gamma}$ is integrable and all the integrals on the right-hand side are well defined. Moreover, by integrability it follows that as $\delta$ goes to $0$ the integrals involving $\delta$ tend to $0$ as well.
\par At last, $\tilde{\Delta}^{(2)}_k(\delta)$ and $\tilde{\Delta}^{(4)}_k(\delta)$ also tend to $0$ since
$$\tilde{\Delta}^{(2)}_k(\delta) \leq 2l_b\delta \pa{\int_{\R^3}\mu_*\langle v_* \rangle^2\:dv_*} \pa{\int_{\R^3}\frac{dv}{\langle v \rangle^{k-2}}}$$ 
and
$$\tilde{\Delta}^{(4)}_k(\delta) \leq \pa{\int_{\R^3}\mu_*\langle v_* \rangle^{2+\gamma}\:dv_*} \pa{\int_{\R^3}\frac{dv}{\langle v \rangle^{k-2-\gamma}}} \cro{\int_{\abs{\cos \theta}\in [1-\delta,1]} b\pa{\cos \theta}\:d\sigma}$$
and $b\pa{\cos \theta}$ is integrable on the sphere. Which concludes the proof of Lemma $\ref{lem:mixingcontrolB2}$.
\end{proof}
\bigskip

We are now able prove the following well-posedness result for $\eqref{f1}$ in the case of Maxwellian diffusion with polynomial weight.

\begin{prop}\label{prop:f1MDpoly}
Let $\Omega$ be a $C^1$ bounded domain. Let $k > 5+\gamma$ and $m_k(v)=\langle v \rangle^k$. Let $f_0$ be in $L^\infty_{x,v}\pa{m_k}$ and $g(t,x,v)$ in $L^\infty_{x,v}\pa{m_k}$. Then there exists $\delta_k$, $\lambda_k(\delta)>0$ such that for any $\delta$ in $(0,\delta_k]$ there exists $C_1$, $\eta_1>0$ such that if
$$ \norm{f_0}_{L^\infty_{x,v}\pa{m_k}}\leq \eta_1 \quad\mbox{and}\quad \norm{g}_{L^\infty_tL^\infty_{x,v}\pa{m_k}}\leq \eta_1,$$
then there exists a solution $f_1$ to 
\begin{equation}\label{f1propMDpoly}
\partial_tf_1 = G_\nu f_1 + B^{(\delta)}_2f_1 + Q(f_1,f_1+g),
\end{equation}
with initial data $f_0$ and satisfying the Maxwellian diffusion boundary conditions $\eqref{MD}$. Moreover, this solution satisfies
$$\forall t\geq 0,\quad \norm{f_1(t)}_{L^\infty_{x,v}\pa{m_k}}\leq C_1 e^{-\lambda_k(\delta) t}\norm{f_0}_{L^\infty_{x,v}\pa{m_k}},$$
and also
$$\lim\limits_{\delta \to 0} \lambda_k(\delta)=\nu_0.$$
The constants $C_1$ and $\eta_1$ are constructive and only depend on $k$, $\delta$ and the kernel of the collision operator.
\end{prop}
\bigskip

\begin{proof}[Proof of Proposition $\ref{prop:f1MDpoly}$]
We closely follow the proof of Proposition $\ref{prop:f1MD}$ in the case of a stretch exponential and we refer to it for most of the details of computations.
\par To build a solution of $\eqref{f1propMDpoly}$ we use the iterative scheme
$$\left\{\begin{array}{l} \disp{h_{l+1} = S_{G_\nu}(t)f_0 + \int_0^t S_{G_\nu}(t-s)\cro{B^{(\delta)}_2 h_{l+1} + Q(h_l,h_l+g)}\:ds } \vspace{2mm}\\ \disp{h_{l+1}(0,x,v) = f_0(x,v) \quad\mbox{and}\quad h_0 = 0.}\end{array}\right.$$
and prove that $\pa{h_l}_{l\in\N}$ is a Cauchy sequence in $L^\infty_{x,v}\pa{m_k}$. Again the well-posed of $h_{l+1}$ would follow from a contraction argument from similar computations.
\par We use the same decomposition as in $\eqref{f1MDdifferenceCauchy}$ with $m(v)$ replaced by $m_k(v)$ :
\begin{equation}\label{decompositionCauchyMDpoly}
\norm{h_{l+1}-h_l}_{L^\infty_{x,v}\pa{m_k}} \leq \max\br{I_1;I_3} + \max\br{I_2; I_4}.
\end{equation}
Since the control of $Q$ and $B^{(\delta)}_2$ also holds in $L^\infty_{x,v}\pa{m_k}$ (see respectively Lemma $\ref{lem:controlQ}$ and Lemma $\ref{lem:controlB2}$) the terms $I_1$, $I_2$ and $I_4$ are estimated in the same way as $\eqref{I1MD}$-$\eqref{I2MD}$-$\eqref{I4MD}$ with $\Delta_{m,\infty}(\delta)$ replaced by $\Delta_{m_k,\infty}(\delta)$ defined in Lemma $\ref{lem:controlB2}$. This gives for any $\eps$ in $(0,1)$
\begin{eqnarray}
I_1 &\leq& e^{-\eps\nu_0 t} \frac{\Delta_{m_k,\infty}(\delta)}{1-\eps}\:\sup\limits_{0\leq s \leq t}\cro{e^{\eps\nu_0 s}\norm{h_{l+1}-h_l}_{L^\infty_{x,v}\pa{m_k}}} \label{I1MDpoly}
\\I_2 &\leq& \frac{C_Q}{1-\eps} e^{-\eps\nu_0 t}\:\sup\limits_{0\leq s\leq t}\cro{e^{\eps\nu_0 s}\norm{h_l-h_{l-1}}_{L^\infty_{x,v}\pa{m_k}}}\label{I2MDpoly}
\\&&\quad\quad\quad\times\cro{\norm{h_l}_{L^\infty_{t,x,v}\pa{m_k}}+\norm{h_{l-1}}_{L^\infty_{t,x,v}\pa{m_k}}+\norm{g}_{L^\infty_{t,x,v}\pa{m_k}}} \nonumber
\\I_4 &\leq &\:C_{m_k}e^{-\eps\nu_0t}\sup\limits_{0\leq s \leq t}\cro{e^{\eps\nu_0 s}\norm{h_{l+1}-h_l}_{L^\infty_{x,v}\pa{m_k}}} \label{I4MDpoly}
\\&&\quad\quad\quad\times\cro{\norm{h_l}_{L^\infty_{t,x,v}\pa{m_k}}+\norm{h_{l-1}}_{L^\infty_{t,x,v}\pa{m_k}}+\norm{g}_{L^\infty_{t,x,v}\pa{m_k}}}.\nonumber
\end{eqnarray}

\bigskip
The main difference lies in $I_3$ where we recall the decomposition $\eqref{decompositionI3}$
\begin{equation}\label{decompositionI3MDpoly}
I_3 \leq I_3^{(1)}+I_3^{(2)}
\end{equation}
with
\begin{eqnarray*}
I_3^{(1)} &=& \sup\limits_{(x,v) \in \Omega\times\R^3}\Big\{\int_0^t c_\mu \mu(v)m_k(v) e^{-\nu(v)(t-s-t_1)}
\\&& \quad\times\sum\limits_{i=1}^{p}\int_{\prod\limits_{j=1}^{p}\br{v_j \cdot n(x_i)>0}}\mathbf{1}_{[t_{i+1},t_i)}(0)\:\abs{B^{(\delta)}_2\pa{h_{l+1}-h_l}}(s,x_i-t_iv_i,v_i)d\Sigma_i(0)\Big\}
\\I_3^{(2)} &=& \sup\limits_{(x,v) \in \Omega\times\R^3}\Big\{\int_0^t  c_\mu \mu(v)m_k(v)e^{-\nu(v)(t-t_1)}
\\&&\quad\times\int_{\prod\limits_{j=1}^{p}\br{v_j \cdot n(x_i)>0}}\mathbf{1}_{t_{p+1}>0}\:\abs{S_{G_\nu}(t_p)\pa{B^{(\delta)}_2\pa{h_{l+1}-h_l}}}(x_p,v_p)d\Sigma_{p}(t_p)\Big\}.
\end{eqnarray*}

\par By definition of $d\Sigma_i(0)$ $\eqref{dSigmai}$ and denoting by $C_k$ the supremum of $c_\mu \mu(v)m_k(v)$ we get
\begin{equation*}
\begin{split}
I_3^{(1)} \leq C_k \int_0^t e^{-\nu_0(t-s)}\sup\limits_{(x,v) \in \Omega\times\R^3}\sum\limits_{i=1}^{p}&\int_{\prod\limits_{\overset{j=1}{j\neq i}}^{p}\br{v_j \cdot n(x_i)>0}}\mathbf{1}_{[t_{i+1},t_i)}(0)
\\&\pa{\int_{\R^3}B_2^{\delta}(h_{l+1}-h_l)(s,x_i,v_i)\abs{v_i}\:dv_i} \:d\sigma_{x_i}ds.
\end{split}
\end{equation*}
We use Lemma $\ref{lem:mixingcontrolB2}$ to estimate the integral in the $i^{th}$ variable and we remind that $d\sigma_{x_i}$ is a probability measure on $\br{v_j\cdot n(x_i)>0}$. This yields
\begin{eqnarray}
I_3^{(1)} &\leq& C_k\tilde{\Delta}_k(\delta)\pa{\int_0^t e^{-\nu_0 (t-s)}\norm{(h_{l+1}-h_l)(s)}_{L^{\infty}_{x,v}\pa{m_k}}\:ds}\nonumber
\\&\leq& e^{-\eps\nu_0 t}C_k \:\frac{\tilde{\Delta}_k(\delta)}{1-\eps}\:\sup\limits_{0\leq s \leq t}\cro{e^{\eps\nu_0 s}\norm{h_{l+1}-h_l}_{L^{\infty}_{x,v}\pa{m_k}}}. \label{I31MDpoly}
\end{eqnarray}

\par The term $I_3^{(2)}$ needs the $L^1_vL^\infty_x\pa{\langle v \rangle^3}$ semigroup theory for $S_{G_\nu}(t)$. Indeed, as in the proof of Proposition $\ref{prop:f1MD}$ we estimate it by
\begin{eqnarray*}
I_3^{(2)} &\leq& C_k \sup\limits_{(x,v) \in \Omega\times\R^3}\Big\{e^{-\nu_0(t-t_1)}\int_0^t e^{-\nu_0(t_1-t_p)}\sum\limits_{i=1}^{p} \int_{\prod\limits_{\overset{j=1}{j\neq i}}^{p}\br{v_j \cdot n(x_i)>0}}\mathbf{1}_{t_p >0}
\\&\:& \quad\quad\quad\pa{\int_{\R^3}\abs{S_{G_\nu}(t_p)\pa{B_2^{\delta}(h_{l+1}-h_l)}\mathbf{1}_{t_1(t_p,x_p,v_p)>0}}\abs{v_p}\:dv_p} \:d\sigma_{x_i}ds\Big\}
\\&\leq& C_k \sup\limits_{(x,v) \in \Omega\times\R^3}\big\{e^{-\nu_0(t-t_1)}\int_0^t ds \:e^{-\nu_0(t_1-t_p)}
\\&\:&\sum\limits_{i=1}^{p} \int_{\prod\limits_{\overset{j=1}{j\neq i}}^{p}\br{v_j \cdot n(x_i)>0}}\mathbf{1}_{t_p >0}\norm{S_{G_\nu}(t_p)\pa{B_2^{\delta}(h_{l+1}-h_l)}\mathbf{1}_{t_1>0}}_{L^1_vL^\infty_x\pa{\langle v \rangle^3}} \:d\sigma_{x_i}ds\Big\}
\end{eqnarray*}
Using Corollary $\ref{cor:SGnuMDcontrolnu-1}$ with $k=3$ and $\nu_0'= (1-\eps')\nu_0$ with $\eps <\eps'<1$ and then applying Lemma $\ref{lem:mixingcontrolB2}$ to obtain
\begin{eqnarray*}
&&\norm{S_{G_\nu}(t_p)\pa{B_2^{\delta}(h_{l+1}-h_l)}\mathbf{1}_{t_1>0}}_{L^1_vL^\infty_x\pa{\langle v \rangle^3}} 
\\&&\quad\quad\quad\quad\quad\leq C_{m_k}e^{-(1-\eps')\nu_0 t_p}\norm{B_2^{\delta}(h_{l+1}-h_l)}_{L^1_vL^\infty_x\pa{\langle v \rangle^2}}
\\&&\quad\quad\quad\quad\quad\leq C_{m_k}\tilde{\Delta}_k(\delta)e^{-(1-\eps')\nu_0 t_p}\norm{h_{l+1}-h_l}_{L^\infty_{x,v}\pa{m_k}}.
\end{eqnarray*}
This estimates allow us to copy the computations made in $\eqref{I32MD}$ and conclude, with $C_{m_k}>0$ a generic constant depending only on $k$
\begin{equation}\label{I32MDpoly}
I_3^{(2)} \leq e^{-\eps\nu_0 t} C_{m_k}\tilde{\Delta}_k(\delta) \sup\limits_{0\leq s \leq t}\cro{e^{\eps\nu_0 s}\norm{h_{l+1}-h_l}_{L^\infty_{x,v}\pa{m_k}}}.
\end{equation}

\par Plugging $\eqref{I31MDpoly}$ and $\eqref{I32MDpoly}$ into $\eqref{decompositionI3MDpoly}$ yields the last estimate
\begin{equation}\label{I3MDpoly}
I_3 \leq e^{-\eps\nu_0t}\:C_{m_k}\tilde{\Delta}_k(\delta)\sup\limits_{0\leq s \leq t}\cro{e^{\eps\nu_0 s}\norm{h_{l+1}-h_l}_{L^\infty_{x,v}\pa{m_k}}}.
\end{equation}

\bigskip
To conclude, we choose $\delta$ small enough such that 
$$\Delta_{m_k,\infty}(\delta) < \Delta_k = \frac{\frac{4}{k-1-\gamma}+1}{2} < 1.$$
Fix $\eps$ in $(0,1)$ such that 
$\eps < 1-\Delta_k$
and finally make $\delta$ even smaller so that in $\eqref{I32MDpoly}$
$$C_{m_k}\tilde{\Delta}_k(\delta) \leq 1 - \frac{\Delta_k}{1-\eps}.$$

\par We gather $\eqref{I1MDpoly}$-$\eqref{I2MDpoly}$-$\eqref{I31MDpoly}$-$\eqref{I4MDpoly}$ and combine them with $\eqref{decompositionCauchyMDpoly}$ 
\begin{equation*}
\begin{split}
&\pa{1-\frac{\Delta_k}{1-\eps}-C_{m_k}\tilde{\Delta}_k(\delta)}\sup\limits_{0\leq s \leq t}\norm{h_{l+1}-h_l}_{L^\infty_{x,v}\pa{m_k}} 
\\&\quad\quad\quad\quad\quad\leq C_{m_k} \sup\limits_{0\leq s \leq t}\cro{e^{\eps\nu_0 s}\norm{h_{l+1}-h_l}_{L^\infty_{x,v}\pa{m_k}}}
\\&\quad\quad\quad\quad\quad\quad\quad\times\cro{\norm{h_l}_{L^\infty_{t,x,v}\pa{m_k}}+\norm{h_{l-1}}_{L^\infty_{t,x,v}\pa{m_k}}+\norm{g}_{L^\infty_{t,x,v}\pa{m_k}}}.
\end{split}
\end{equation*}
Since the constant on the left-hand side is positive we conclude with exactly the same arguments as in the end of the proof of Proposition $\ref{prop:f1SR}$ or Proposition $\ref{prop:f1MD}$.
\end{proof}
\bigskip


\subsection{Existence and exponential decay for equation $\eqref{f2}$ in $L^\infty_{x,v}\pa{\langle v \rangle^\beta\mu^{-1/2}}$}\label{subsec:f2}

In this section we establish the well-posedness and the exponential decay of $\eqref{f2}$ in $L^\infty_{x,v}\pa{\langle v \rangle^\beta\mu^{-1/2}}$, with $\beta$ such that Theorem $\ref{theo:semigroupLinfty}$ holds.

\bigskip
\begin{prop}\label{prop:f2}
Let $\Omega$ be an analytic strictly convex domain (resp. a $C^1$ bounded domain)and let $0<\lambda_G'<\lambda_G$ (defined by Theorem $\ref{theo:semigroupLinfty}$). Let $m=e^{\kappa\abs{v}^\alpha}$ with $\kappa >0$ and $\alpha$ in $(0,2)$ or $m=\langle v \rangle^k$ with $k>5+\gamma$. Then there exists $\eta_2>0$ such that if $g(t,x,v)$ is in $L^\infty_{x,v}\pa{m}$ with
$$\forall t\geq 0, \quad \norm{g(t)}_{L^\infty_{x,v}\pa{m}}\leq \eta_2e^{-\lambda_G t}.$$
Then there exists a solution $f_2$ in $L^\infty_{x,v}\pa{\langle v \rangle^\beta\mu^{-1/2}}$ to 
\begin{equation}\label{f2prop}
\partial_tf_2 = G f_2  + Q(f_2,f_2) + A^{(\delta)}g,
\end{equation}
with zero as initial data and satisfying the specular reflections (resp. Maxwellian diffusion) boundary conditions. Moreover, if we assume $\Pi_G(f_2+g) = 0$ then there exists $C_2>0$ such that
$$\forall t\geq 0,\quad \norm{f_2(t)}_{L^\infty_{x,v}\pa{\langle v \rangle^\beta\mu^{-1/2}}}\leq C_2 \eta_2 e^{-\lambda_G' t}.$$
$\Pi_G$ is the projection on the kernel of $G$ and depends on the boundary conditions (see $\eqref{PiGSR}$-$\eqref{PiGMD}$).The constants $\eta_2$ and $C_2$ are constructive and only depend on $\lambda_G'$, $k$, $q$, $\delta$ and the kernel of the collision operator. 
\end{prop}
\bigskip

\begin{proof}[Proof of Proposition $\ref{prop:f2}$]
We start by noticing that the Cauchy problem for 
$$\partial_t f = Gf + Q(f,f)$$
has been solved in $L^\infty_{x,v}\pa{\langle v \rangle^\beta\mu^{-1/2}}$ for small initial data in \cite{Gu6} Section $5$, both for specular and diffusive boundary conditions. The addition of a mere source term $A^{(\delta)}g$ for a null initial data is handled the same way and we therefore have the existence of $f_2$ solution to $\eqref{f2prop}$ in $L^\infty_{x,v}\pa{\langle v \rangle^\beta\mu^{-1/2}}$.

\bigskip
We suppose that $\Pi_G (f_2+g)=0$, in other words we ask for $f_2+g$ to have the appropriate conservation laws depending on the boundary conditions. We would like to apply the $L^\infty$ theory for $S_G$ given by Theorem $\ref{theo:semigroupLinfty}$ but it is only applicable in the space of functions in $L^\infty_{x,v}\pa{\langle v \rangle^\beta\mu^{-1/2}}$ satisfying the respective boundary conditions. We thus need to independently study $\Pi_G(f_2)$ and $\Pi_G^\bot(f_2)$.

\bigskip
\textbf{Study of the projection $\Pi_G(f_2)$.} Since $\Pi_G (f_2+g)=0$, we have that $\Pi_G(f_2) = -\Pi_G(g)$. By Definition  $\ref{PiGSR}$ for specular reflections or Definition $\ref{PiGMD}$ for Maxwellian diffusion we have the following form for $\Pi_G(g)$:
$$\Pi_G (g) = \sum\limits_{i=0}^{d+1}c_i\pa{\int_{\Omega\times\R^3}g(t,x,v)\phi_i(v)\:dxdv}\phi_i(v)\mu(v)$$
with $c_i$ is either $0$ or $1$ and $\phi_i$ is given by $\eqref{piL}$. It follows that
$$\norm{\Pi_G(f_2)}_{L^\infty_{x,v}\pa{\langle v \rangle^\beta\mu^{-1/2}}} \leq \sum\limits_{i=0}^{d+1}\abs{\int_{\Omega\times\R^3}g(t,x,v)\phi_i(v)\:dxdv}\sup\limits_{v\in\R^3}\pa{\langle v \rangle^\beta\phi\mu^{1/2}}.$$
Since $k>5+\gamma$ and $\phi_i$ is a polynomial in $v$ of order $0$, $1$ or $2$ it follows
$$\abs{\int_{\Omega\times\R^3}g(t,x,v)\phi_i(v)\:dxdv} \leq \abs{\Omega}\pa{\int_{\R^3}\langle v \rangle^{-k}\phi_i(v)\:dv}\norm{g}_{L^\infty_{x,v}\pa{m}}.$$
As a conclusion, there exists $C_\Pi >0$ such that
\begin{equation}\label{PiGf2}
\forall t\geq 0, \quad \norm{\Pi_G(f_2)}_{L^\infty_{x,v}\pa{\langle v \rangle^\beta\mu^{-1/2}}} \leq C_\Pi \norm{g}_{L^\infty_{x,v}\pa{m}}\leq C_\Pi \eta_2 e^{-\lambda_Gt}.
\end{equation}

\bigskip
\textbf{Study of the orthogonal part of $f_2$.} By definition we have that $\Pi_G(G(f_2))=0$ and $G(f_2) = G(\Pi_G^\bot(f_2))$. Thanks to the orthogonality property of $Q$, given by Lemma $\ref{lem:controlQ}$, $\Pi_G(Q(f_2,f_2))=0$. Therefore, $F_2 = \Pi_G^\bot(f_2)$ satisfies the following differential equation
$$\partial_t F_2 = G(F_2) + Q(f_2,f_2) + \Pi_G^\bot\pa{A^{(\delta)}g}.$$
Every term in the latter equation satisfies the conservation laws associated to the boundary conditions. We can therefore use Theorem $\ref{theo:semigroupLinfty}$ and have the following Duhamel representation for $F_2$ almost everywhere
\begin{equation}\label{F2}
F_2(t,x,v) = \int_0^t S_G(t-s)Q(f_2,f_2)(x,v)\:ds + \int_0^tS_G(t-s)\Pi_G^\bot\pa{A^{(\delta)}g}\:ds
\end{equation}

\par The first term on the right-hand side of $\eqref{F2}$ is dealt with by Theorem $\ref{theo:semigroupLinfty}$ with $\lambda_G'<\lambda_G$.
\begin{equation*}
\begin{split}
&\norm{\int_0^t S_G(t-s)Q(f_2,f_2)(x,v)\:ds}_{L^\infty_{x,v}\pa{\langle v \rangle^\beta\mu^{-1/2}}} 
\\&\quad\quad\quad\quad\quad\leq  C_G e^{-\lambda_G' t}\sup\limits_{s\in [0,t]}\pa{e^{\lambda_G's} \norm{f_2(s)}^2_{L^\infty_{x,v}\pa{\langle v \rangle^\beta\mu^{-1/2}}}}.
\end{split}
\end{equation*}
Then we use the fact that $f_2 = F_2 +\Pi_G(f_2)$ together with the exponential decay of $\Pi_G(f_2)$ $\eqref{PiGf2}$. This yields
\begin{equation}\label{PiGbotf2Q}
\begin{split}
&\norm{\int_0^t S_G(t-s)Q(f_2,f_2)(x,v)\:ds}_{L^\infty_{x,v}\pa{\langle v \rangle^\beta\mu^{-1/2}}} 
\\&\quad\quad\quad\leq 2C_Ge^{-\lambda'_G t}\cro{\pa{\sup\limits_{s\in [0,t]}\pa{e^{\lambda'_Gs} \norm{F_2(s)}_{L^\infty_{x,v}\pa{\langle v \rangle^\beta\mu^{-1/2}}}}}^2+ C^2_\Pi \eta_2^2}.
\end{split}
\end{equation}

\bigskip
For the second term on the right-hand side of $\eqref{F2}$ we use Theorem $\ref{theo:semigroupLinfty}$ to get
\begin{equation*}
\begin{split}
&\norm{\int_0^t S_G(t-s)\Pi_G^\bot\pa{A^{(\delta)}g}\:ds}_{L^\infty_{x,v}\pa{\langle v \rangle^\beta\mu^{-1/2}}} 
\\&\quad\quad\quad\quad\quad\leq C_G\int_0^t e^{-\lambda_G(t-s)}\norm{\Pi_G^\bot\pa{A^{(\delta)}g}(s)}_{L^\infty_{x,v}\pa{\langle v \rangle^\beta\mu^{-1/2}}}\:ds.
\end{split}
\end{equation*}
Here again we can bound the norm of $\Pi_G^\bot\pa{A^{(\delta)}g}$ by the norm of $A^{(\delta)}g$ which is itself bounded by Lemma $\ref{lem:controlA}$. This yields
\begin{eqnarray*}
\norm{\Pi_G^\bot\pa{A^{(\delta)}g}(s)}_{L^\infty_{x,v}\pa{\langle v \rangle^\beta\mu^{-1/2}}} &\leq& C_\Pi C_A \norm{g(s)}_{L^\infty_{x,v}\pa{m}}
\\&\leq& \eta_2 C_\Pi C_A e^{-\lambda_G s}.
\end{eqnarray*}
Hence, since $\lambda_G \leq \lambda'_G$,
\begin{eqnarray}
\norm{\int_0^t S_G(t-s)\Pi_G^\bot\pa{A^{(\delta)}g}\:ds}_{L^\infty_{x,v}\pa{\langle v \rangle^\beta\mu^{-1/2}}} &\leq& \eta_2 C_\Pi C_A C_G  t e^{-\lambda_G t} \nonumber
\\&\leq& \eta_2C_{\Pi^\bot}e^{-\lambda_G't},\label{PiGbotf2A}
\end{eqnarray}
where $C_{\Pi^\bot}>0$ is a constant depending on $\lambda_G'$.

\bigskip
Plugging $\eqref{PiGbotf2Q}$ and $\eqref{PiGbotf2A}$ into $\eqref{F2}$ yields, with $C_2$ denoting any positive constant independent of $\eta_2$,
$$e^{\lambda_G't}\norm{F_2(t)}_{L^\infty_{x,v}\pa{\langle v \rangle^\beta\mu^{-1/2}}} \leq  C_2\cro{ \eta_2^2 + \eta_2 +  \pa{\sup\limits_{s\in [0,t]}\pa{e^{\lambda'_Gs} \norm{F_2(s)}_{L^\infty_{x,v}\pa{\langle v \rangle^\beta\mu^{-1/2}}}}}^2}.$$

\par At $t=0$, $F_2=0$ and therefore we can define
$$t_M = \sup\br{t\geq 0:\quad e^{\lambda_G't}\norm{F_2(t)}_{L^\infty_{x,v}\pa{\langle v \rangle^\beta\mu^{-1/2}}}\leq (C_2+2)\eta_2}.$$
Suppose that $t_M < +\infty$ then we have that
$$e^{\lambda_G't_M}\norm{F_2(t_M)}_{L^\infty_{x,v}\pa{\langle v \rangle^\beta\mu^{-1/2}}} \leq C_2\eta_2 + C_2\pa{(C_2+2)^2+1}\eta_2^2$$
and if $\eta_2$ is small enough:
$$e^{\lambda_G't_M}\norm{F_2(t_M)}_{L^\infty_{x,v}\pa{\langle v \rangle^\beta\mu^{-1/2}}} \leq (C_2+1)\eta_2.$$
Therefore if $\eta_2$ is small enough we reach a contradiction by definition of $t_M$, which implies $t_M = +\infty$ and
\begin{equation}\label{PiGbotf2}
\forall t\geq 0, \quad e^{\lambda_G't}\norm{\Pi_G^\bot(f_2)}_{L^\infty_{x,v}\pa{\langle v \rangle^\beta\mu^{-1/2}}} \leq (C_2+2)\eta_2.
\end{equation}

\bigskip
To conclude the proof we simply gather the control of $\Pi_G(f_2)$ of $\eqref{PiGf2}$ and the control of the orthogonal part $\eqref{PiGbotf2}$.
\end{proof}
\bigskip

%% file: mainproofs.tex
\section{Cauchy theory for the full Boltzmann equation}\label{sec:mainproofs}

This section is dedicated to the proof of Theorem $\ref{theo:cauchySR}$ and of Theorem $\ref{theo:cauchyMD}$. We tackle each of the issues: existence and exponential decay, uniqueness, continuity and positivity separately.
\bigskip


\subsection{Existence and exponential decay}\label{subsec:existenceexpodecay}

The existence and exponential trend to equilibrium of $F$ solution to the full Boltzmann equation $\eqref{BE}$ near to equilibrium $F=\mu + f$ is equivalent to the existence and exponential decay of $f$ solution to the perturbed Boltzmann equation $\eqref{perturbedBE}$. The latter directly follows from Proposition $\ref{prop:f1SR}$, Proposition $\ref{prop:f1MD}$ and Proposition $\ref{prop:f2}$. Indeed, we consider the following scheme.
\par Define $f_1^{(0)}(t,x,v) = f^{(0)}_2(t,x,v) = 0$ and the iterative process
\begin{eqnarray*}
\partial_t f^{(l+1)}_1 &=& B^{(\delta)} f^{(l+1)}_1 + Q(f_1^{(l+1)},f_1^{(l+1)}+f^{(l)}_2) \quad\mbox{and}\quad f^{(l+1)}_1(0,x,v)=f_0(x,v),
\\\partial_t f^{(l+1)}_2 &=& G f^{(l+1)}_2 + Q(f^{(l+1)}_2,f^{(l+1)}_2) + A^{(\delta)}f^{(l)}_1 \quad\mbox{and}\quad f^{(l+1)}_2(0,x,v)=0
\end{eqnarray*}
with either specular reflections or Maxwellian diffusion boundary conditions and the additional condition $\Pi_G (f_2^{(l+1)}+f_1^{(l)})=0$.

\bigskip
Suppose that $\norm{f_0}_{L^\infty_{x,v}\pa{m}}$ is smaller than $\eta_0$ with $\eta_0$ such that
$$C_1C_2\eta_0 \leq \eta_1 \quad\mbox{and}\quad C_1\eta_0 \leq \eta_2,$$
where the constants $C_1$, $\eta_1$ are defined in Proposition $\ref{prop:f1SR}$ or Proposition $\ref{prop:f1MD}$ (depending on the boundary conditions) and $C_2$, $\eta_2$ in Proposition $\ref{prop:f2}$.

\bigskip
We define $\lambda = \min\br{\lambda_G,\lambda_m(\delta)}$. Fix $\lambda'$ in $[0,\lambda]$.
\par By induction we shall prove that for all $l$, $f^{(l)}_1$ and $f^{(l)}_2$ are well-defined and satisfy
\begin{eqnarray*}
\forall t\geq 0, \quad &&\norm{f^{(l)}_1}_{L^\infty_{x,v}\pa{m}} \leq C_1 e^{-\lambda t}\norm{f_0}_{L^\infty_{x,v}\pa{m}}
\\ &&\norm{f^{(l)}_2}_{L^\infty_{x,v}\pa{\langle v \rangle^\beta\mu^{-1/2}}} \leq C_1C_2 e^{-\lambda't}\norm{f_0}_{L^\infty_{x,v}\pa{m}},
\end{eqnarray*}
where $\beta$ is such that Theorem $\ref{theo:semigroupLinfty}$ holds.
\par If the latter inequalities hold at rank $l$ then by definition of $\eta_0$ we have that
$$\norm{f^{(l)}_2}_{L^\infty_tL^\infty_{x,v}\pa{\langle v \rangle^\beta\mu^{-1/2}}} \leq \eta_1 \quad\mbox{and}\quad  \norm{f^{(l)}_1}_{L^\infty_tL^\infty_{x,v}\pa{m}} \leq \eta_2 e^{-\lambda t}$$
and by Proposition $\ref{prop:f1SR}$, Proposition $\ref{prop:f1MD}$ and Proposition $\ref{prop:f2}$ we can therefore construct $f_1^{(l+1)}$ and $f_2^{(l+1)}$. Moreover, these functions satisfy for all $t\geq 0$,
\begin{eqnarray*}
&&\norm{f^{(l+1)}_1}_{L^\infty_{x,v}\pa{m}} \leq C_1 e^{-\lambda t}\norm{f_0}_{L^\infty_{x,v}\pa{m}},
\\ &&\norm{f^{(l+1)}_2}_{L^\infty_{x,v}\pa{\langle v \rangle^\beta\mu^{-1/2}}} \leq C_2 e^{-\lambda 't}\norm{f^{(l)}_1}_{L^\infty_tL^\infty_{x,v}\pa{m}} \leq C_1C_2e^{-\lambda't}\norm{f_0}_{L^\infty_{x,v}\pa{m}}.
\end{eqnarray*}

\bigskip
We thus derive the weak-* convergence of $\pa{f_1^{(l)}}_{l \in \N}$ and $\pa{f_2^{(l)}}_{l \in \N}$ (up to subsequences) towards $f_1$ and $f_2$ solutions of the system of equations $\eqref{f1}$-$\eqref{f2}$. Therefore $f=f_1+f_2$ is a solution to the perturbed Boltzmann equation  $\eqref{perturbedBE}$ and satisfies the desired exponential decay.
\bigskip


\subsection{Uniqueness of solutions}\label{subsec:uniqueness}

Like the results about uniqueness obtained in \cite{Gu6} in $L^\infty_{x,v}\pa{\langle v \rangle^\beta \mu^{-1/2}}$, the uniqueness results given in Theorem $\ref{theo:cauchySR}$ and of Theorem $\ref{theo:cauchyMD}$ only apply in a perturbative regime. In other words it states the uniqueness of functions of the specific form $F= \mu + f$. This allows to use most of the computations made in previous sections.
\par More precisely, we fix boundary conditions (either specular or diffusive) and we consider $f_0$ such that
$$\norm{f_0}_{L^\infty_{x,v}\pa{m}}\leq \eta_0$$
with $\eta_0$ small enough such that we can construct (see previous subsection) $F= \mu + f$ a solution to the full Boltzmann equation. Note that we have exponential decay for $f$ in the following form
\begin{equation}\label{expodecayuniqueness}
\exists C_m>0, \: \forall t\geq 0,\quad \norm{f(t)}_{L^\infty_{x,v}\pa{m}} \leq C_m e^{-\lambda_m t}\norm{f_0}_{L^\infty_{x,v}\pa{m}}.
\end{equation}
We are about to prove that any other solution to the full Boltzmann equation of the form $H =\mu + h$ with $h_0=f_0$ and satisfying the boundary conditions must be $F$ itself on condition that $\eta_0$ is small enough.

\bigskip
Consider $H = \mu + h$ to be another solution to the Boltzmann equation with the same boundary conditions as $F$ and the same initial data then $f-h$ satisfies
$$\partial_t \pa{f-h} + = G\pa{f-h} + Q\pa{f+h,f-h},$$
with the same boundary conditions and zero as initial data. If 
\begin{equation}\label{etauniqueness}
\sup\limits_{0\leq t \leq T_0}\norm{f+h}_{L^\infty_{x,v}\pa{m}}\leq \eta
\end{equation}
where $\eta$ is small enough then we can use exactly the same computations as in Section $\ref{sec:extensionSRMD}$ to prove that for a given initial data $g_0 \in L^\infty_{x,v}\pa{m}$ there exists a solution in $L^\infty_{x,v}\pa{m}$ to
\begin{equation}\label{eqdiffuniqueness}
\partial_t g + = G\pa{g} + Q\pa{f+h,g},
\end{equation}
with the required boundary condition (the smallness assumption on $f+h$ playing the same role as the smallness of $\Delta_{m,\infty}(\delta)$). 
\par Moreover, if $g_0$ is small enough to fit the computations in Section $\ref{sec:extensionSRMD}$ we have an exponential decay for $g$ and in particular
$$\exists C_m>0,\:\forall t\geq 0, \quad \norm{g(t)}_{L^\infty_{x,v}\pa{m}}\leq  C_m \norm{g_0}_{L^\infty_{x,v}\pa{m}}.$$
The latter inequality yields uniqueness for $g$ for small initial data $g_0$.

\bigskip
The uniqueness $f=h$ follows from a bootstrap argument. Consider $\eta_0$ such that
$$\eta_0\leq \frac{\eta}{4\max\br{1,C_m}}$$
with $\eta$ defined in $\eqref{etauniqueness}$ and $C_m$ defined by $\eqref{expodecayuniqueness}$. Define
$$T_0 = \sup\br{T >0,\quad \norm{f(T)}_{L^\infty_{x,v}\pa{m}}\leq \eta/2\quad\mbox{and}\quad\norm{h(T)}_{L^\infty_{x,v}\pa{m}}\leq \eta/2}.$$
Suppose $T_0<+\infty$, then we have that $\eqref{etauniqueness}$ holds on $[0,T_0]$ and therefore $f-h$ is the unique solution to $\eqref{eqdiffuniqueness}$ with initial value $0$ and thus almost everywhere in $\Omega\times\R^3$
$$\forall t \in [0,T_0],\quad f(t,x,v)=h(t,x,v).$$
Thanks to the exponential decay $\eqref{expodecayuniqueness}$ satisfied by $f$ we have at $T_0$
$$\norm{f(T_0)}_{L^\infty_{x,v}\pa{m}} \leq C_m \eta_0 \leq \frac{\eta}{4}$$
and
$$\norm{h(T_0)}_{L^\infty_{x,v}\pa{m}} \leq C_m \eta_0 \leq \frac{\eta}{4}.$$
This contradicts the definition of $T_0$ and therefore we must have $T_0 = +\infty$ and $f(t,x,v)=h(t,x,v)$ for all $t\geq 0$ and almost every $(x,v)$ in $\Omega\times\R^3$. This concludes the proof of uniqueness in the perturbative regime.
\bigskip


\subsection{Continuity of solutions away from the grazing set}\label{subsec:continuity}

The continuity of solutions away from the grazing set have already been studied in \cite{Gu6}, both for specular reflections and Maxwellian diffusion in convex domains, and in \cite{EGKM} for Maxwellian diffusion with more general bounded domains. We prove here that their results apply in our present work.
\par Obviously, the continuity of $F=\mu+f$ is equivalent to the one of $f$, which we tackle here.

\bigskip
We recall \cite{Gu6} Lemma $21$ with our notations.
\begin{lemma}\label{lem:continuitySR}
Let $\Omega$ be a $C^2$ strictly convex in the sense $\eqref{strictlyconvex}$ bounded domain. Let $f_0$ be continuous on $\bar{\Omega}\times\R^3 - \Lambda_0$ and satisfying the specular reflections condition at the boundary. At last, let $q(t,x,v)$ be continuous in the interior of $[0,+\infty)\times\Omega\times\R^3$ with
$$\sup\limits_{[0,+\infty)\times\Omega\times\R^3}\abs{\frac{q(t,x,v)}{\nu(v)}} < +\infty.$$
Then the solution to
$$\partial_t f = G_\nu(f) + q(t,x,v)$$
with initial data $f_0$ satisfying the specular reflections boundary conditions is continuous on $[0,+\infty)\times \pa{\bar{\Omega}\times\R^3 - \Lambda_0}$.
\end{lemma}

\bigskip
Thanks to the previous proof of uniqueness, $f$ is the limit of $f_1^{(l)} + f_2^{(l)}$ where the two sequences have been defined in Subsection $\ref{subsec:existenceexpodecay}$ by
\begin{eqnarray*}
\partial_t f^{(l+1)}_1 &=& G_\nu f^{(l+1)}_1 + B^{(\delta)}_2 f^{(l+1)}_1 + Q(f_1^{(l+1)},f_1^{(l+1)}+f^{(l)}_2),
\\\partial_t f^{(l+1)}_2 &=& G_\nu f^{(l+1)}_2 + (A^{(\delta)}+B^{(\delta)}_2)f^{(l+1)}_2 + Q(f^{(l+1)}_2,f^{(l+1)}_2) + A^{(\delta)}f^{(l)}_1
\end{eqnarray*}
with $f^{(l+1)}_1(0,x,v)=f_0(x,v)$ and $f^{(l+1)}_2(0,x,v)=0$ and specular reflections boundary condition.
\par We use the method of \cite{Gu6} Section $1$ and approximate these solutions by respectively $f^{(l,l')}_1$ and $f^{(l,l')}_2$ with same initial data and boundary conditions and satisfying both  $f^{(l,0)}_1=f^{(l,0)}_2 =0$ and the inductive property
\begin{eqnarray*}
\partial_t f^{(l+1,l'+1)}_1 &=& G_\nu f^{(l+1,l'+1)}_1 + B^{(\delta)}_2 f^{(l+1,l')}_1 + Q(f_1^{(l+1,l')},f_1^{(l+1,l')}+f^{(l)}_2),
\\\partial_t f^{(l+1,l'+1)}_2 &=& G_\nu f^{(l+1,l'+1)}_2 + (A^{(\delta)}+B^{(\delta)}_2)f^{(l+1,l')}_2 + Q(f^{(l+1,l')}_2,f^{(l+1,l')}_2) + A^{(\delta)}f^{(l)}_1.
\end{eqnarray*}
Note that all the functions involved here are in $L^\infty_{x,v}\pa{m}$.

\bigskip
The proof of continuity is then done by induction on $l$.
\par Suppose that $f^{(l)}_1$ and $f^{(l)}_2$ are continuous on $[0,+\infty)\times \pa{\bar{\Omega}\times\R^3 - \Lambda_0}$ and that $f^{(l+1,l')}_1$ and $f^{(l+1,l')}_2$ as well. Then we can easily apply Lemma $\ref{lem:continuitySR}$ to $ f^{(l+1,l'+1)}_1$ and $f^{(l+1,l'+1)}_2$ thanks to the control on $A^{(\delta)}$ (Lemma $\ref{lem:controlA}$), on $B^{(\delta)}_2$ (Lemma $\ref{lem:controlB2}$) and on $Q$ (Lemma $\ref{lem:controlQ}$) and the fact that $f^{(l)}_1$, $f^{(l)}_2$, $f^{(l+1,l')}_1$ and $f^{(l+1,l')}_2$ are all in $L^\infty_{x,v}\pa{m}$. Therefore $ f^{(l+1,l'+1)}_1$ and $f^{(l+1,l'+1)}_2$ are also continuous on $[0,+\infty)\times \pa{\bar{\Omega}\times\R^3 - \Lambda_0}$.
\par To conclude, same computations as in Subsection $\ref{subsec:f1}$ and Subsection $\ref{subsec:f2}$ shows that $\pa{f^{(l+1,l')}_1}_{l'\in\N}$ and $\pa{f^{(l+1,l')}_2}_{l'\in\N}$ are Cauchy sequences in $L^\infty_{x,v}\pa{m}$ and therefore their respective limits, $f^{(l+1)}_1$ and $f^{(l+1)}_2$, are also continuous away from the grazing set. Which concludes the induction
\par Thanks to Subsection $\ref{subsec:f1}$ and Subsection $\ref{subsec:f2}$ we also know that $\pa{f^{(l)}_1}_{l\in\N}$ and $\pa{f^{(l)}_2}_{l'\in\N}$ are Cauchy sequences in $L^\infty_{x,v}\pa{m}$ and hence their respective limits are also continuous away from the grazing set. This concludes the fact that $f_1+f_2$ and therefore $F=\mu + f_1+f_2$ are continuous on  $[0,+\infty)\times \pa{\bar{\Omega}\times\R^3 - \Lambda_0}$ in the case of specular reflections.

\bigskip
The case of Maxwellian diffusion boundary condition is dealt with thanks to similar arguments, starting from the continuity Lemma $26$ in \cite{Gu6} which is equivalent to Lemma $\ref{lem:continuitySR}$ for diffusive boundary in the case of a convex bounded domain or from Proposition $6.1$ in \cite{EGKM} for more general $C^1$ bounded domains.
\bigskip


\subsection{Positivity of solutions}\label{subsec:positivity}

The positivity of the solutions to the Boltzmann equation $\eqref{BE}$ follows from two recent results by the author \cite{Bri2}\cite{Bri5}. The latter articles give constructive \textit{a priori} maxwellian lower bounds on the solutions to the Boltzmann equation  in $C^2$ convex bounded domains with specular reflections boundary conditions \cite{Bri2} and Maxwellian diffusion boundary conditions \cite{Bri5}.
\par More precisely, in both cases, the following property holds when the collision kernel describes a hard potential with Grad's angular cutoff (which is true for the collision kernels considered in the present work). If a solution $F$ to the Boltzmann equation on $[0,T_{max})$, $T_{max}$ can be infinity, satisfies
\begin{enumerate}
\item[(i)] $F_0$ is a non-negative function with positive mass 
$$M = \int_{\Omega\times\R^3}F(t,x,v)\:dxdv >0,$$
\item[(ii)] $F$ is continuous on $[0,T_{max}) \times \pa{\bar{\Omega} \times\R^3-\Lambda_0}$, in other words continuous away from the grazing set,
\item[(iii)] $F$ has uniformly bounded local energy
$$E_F =\sup\limits_{(t,x) \in [0,T_{max})\times\R^3}\int_{\R^3}\abs{v}^2 F(t,x,v)\:dv <+\infty,$$
\end{enumerate}
then for all $\tau \in (0,T_{max})$ there exists $\rho_\tau$, $\theta_\tau >0$ depending only on $M$, $E_F$, $\tau$ and the collision kernel such that almost everywhere
$$\forall t \in [\tau,T_{max}) ,\:\forall (x,v)\in\bar{\Omega}\times\R^3, \quad F(t,x,v) \geq \frac{\rho_\tau}{\pa{2\pi\theta_\tau}^{3/2}}e^{-\frac{\abs{v}^2}{}2\theta_\tau}.$$

\bigskip
In the present work we constructed solutions to the Boltzmann equation in $L^\infty_{x,v}\pa{m}$ of the form $F=\mu +f$ with $F_0 = \mu + f_0 \geq 0$ satisfying the conservation laws associated with the boundary conditions. Therefore $F$ preserves the total mass so in our case $M =1 >0$ so that point $(i)$ is satisfied. Point $(ii)$ is exactly what we proved in the previous subsection. Finally, since the solution $F$ is in $L^\infty_{x,v}\pa{m}$ with exponential trend to equilibrium it follows
$$\int_{\R^3}\abs{v^2}F(t,x,v)\:dv \leq \norm{F(t)}_{L^\infty_{x,v}\pa{m}}\int_{\R^3}\frac{\abs{v}^2}{m(v)^2}\:dv \leq C_m \norm{F_0}_{L^\infty_{x,v}\pa{m}}$$
and point $(iii)$ is also satisfied. The positivity of $F$ therefore follows from the lower bound property described above.
\bigskip